\documentclass{article}

\usepackage{microtype}
\usepackage{graphicx}
\usepackage{subfigure}
\usepackage{booktabs} 

\usepackage{hyperref}

\usepackage[algo2e,ruled,vlined]{algorithm2e}
\usepackage{algorithmic}
\usepackage{multirow}
\usepackage[accepted]{icml2023}

\usepackage{amsmath}
\usepackage{amssymb}
\usepackage{mathtools}
\usepackage{amsthm}
\usepackage{threeparttable}
\usepackage{ulem}
\usepackage{natbib}
\setcitestyle{numbers,square}
\usepackage[capitalize,noabbrev]{cleveref}

\theoremstyle{plain}
\newtheorem{theorem}{Theorem}[section]
\newtheorem{proposition}[theorem]{Proposition}
\newtheorem{lemma}[theorem]{Lemma}

\theoremstyle{definition}

\theoremstyle{remark}

\usepackage[textsize=tiny]{todonotes}

\icmltitlerunning{Reducing Popularity Bias in Recommender Systems through AUC-Optimal Negative Sampling}

\begin{document}

\onecolumn
\icmltitle{Reducing Popularity Bias in Recommender Systems through AUC-Optimal Negative Sampling}




\begin{icmlauthorlist}
\icmlauthor{Bin Liu}{sch}
\icmlauthor{Eerjia Chen}{sch}
\icmlauthor{Bang Wang*}{sch}
\end{icmlauthorlist}

\icmlaffiliation{sch}{School of Electronic Information and Communications, Huazhong University of Science and Technology (HUST), Wuhan, China}

\icmlcorrespondingauthor{Bang Wang}{wangbang@hust.edu.cn}


\vskip 0.3in




\begin{abstract}
Popularity bias is a persistent issue associated with recommendation systems, posing challenges to both fairness and efficiency. Existing literature widely acknowledges that reducing popularity bias often requires sacrificing recommendation accuracy. In this paper, we challenge this commonly held belief. Our analysis under general bias-variance decomposition framework shows that reducing bias can actually lead to improved model performance under certain conditions. To achieve this win-win situation, we propose to intervene in model training through negative sampling thereby modifying model predictions. Specifically, we provide an optimal negative sampling rule that maximizes partial AUC to preserve the accuracy of any given model, while correcting sample information and prior information to reduce popularity bias in a flexible and principled way. Our experimental results on real-world datasets demonstrate the superiority of our approach in improving recommendation performance and reducing popularity bias~\footnote{School of Electronic Information and Communications, Huazhong University of Science and Technology (HUST), Wuhan, China. \\*Corresponding to: wangbang@hust.edu.cn \\Bin Liu and Erjia Chen contributed equally.}.
\end{abstract}
\section{Introduction}
Recommendation systems have achieved remarkable results in application scenarios such as online shopping and intelligent search. In recent years, research on recommendation systems has primarily focused on designing better models to mine user behavior data and improve recommendation quality. However, bias is a persistent issue associated with recommendation systems and poses challenges to both fairness and efficiency~\cite{Schnabel2016:PMLR,Chen2023:TOIS,Deldjoo2023:UMUAI}. There are two main types of biases: \textit{model-independent bias}, which is and mainly derived from the dataset such as selection bias~\cite{Ovaisi2020:WWW,chen2021autodebias}; and \textit{model-dependent bias}, which is dependent on the model, such as popularity bias~\cite{Abdollahpouri:2019:AAAI, Abdollahpouri:2019:arxiv}, position bias~\cite{Guo2019:RecSys} and exposure bias~\cite{Gupta2021:PMLR}, as the fundamental idea of collaborative filtering algorithms is to learn individual preferences based on group tastes. Due to the complex and diverse nature of biases, reducing recommendation bias remains a exciting challenge.

Popularity bias is a typical outcome of combined effects~\cite{Guo2019:RecSys,Ovaisi2020:WWW,Gupta2021:PMLR,Zhu:2021:SIGKDD} of model-dependent and model-independent biases. Users' selection bias caused a long-tail distribution of item popularity in the dataset, which is further amplified by collaborative filtering algorithms. This situation creates feedback loop~\cite{Zhu:2021:SIGKDD} that degrades the performance and fairness of recommendations, resulting in popular items being excessively prioritized and the preferences of niche groups being ignored. 

Reducing popularity bias often poses a dilemma: excessively reducing the recommendation of popular items may result in underestimating users' preferences for them, despite the fact that popular items may indeed be their preferred choices. The crucial issue lies in accurately estimating the preferences for popular items, in order to align them with users' actual preferences. Typically, this problem can be approached from three perspectives: (i) post-ranking methods for modifying predicted preference values ~\cite{Steck:2018:RecSys,Abdollahpouri:2019:arxiv,Steck:2019:arxiv,Zhu:2021:KDD,Zhu:2021:WSDM,Liu:2017:SIGIR}, (ii) empirical risk rewriting methods for re-weighting samples ~\cite{Schnabel:2016:PMLR,Bottou:2013:JMLR,Gruson:2019:WSDM,Joachims:2017:IJCAI,Yang:2018:RecSys,Saito:2020:WSDM}, and (iii) causal inference methods for disentangling the causal relationship between item popularity and exposure rate ~\cite{Bonner:2018:RecSys,Zhang:2021:SIGIR,Wang:2021:SIGDKK,Wei:2021:SIGDKK,Zheng:2021:WWW,Liu:2020:SIGIR,Zhan:2022:SIGKDD,Xu:2021:arXiv}. Nevertheless, existing methods have the following drawbacks: (i) model and data bias are often treated separately, whereas in reality, they tend to interact with each other and form feedback loops; (ii) complex models are often introduced, which degrade the model's generalization performance.

In this paper, we propose a flexible and lightweight approach to reduce popularity bias without compromising the recommendation accuracy of the given model or introducing additional complexity. The basic idea is to intervene in model training by using negative sampling as a pre-step to compute contrastive loss, thereby modifying model prediction. Specifically, we correct sample information that targeted on model-dependent bias, and prior information that targeted on model-independent bias for guiding negative sampling, which maximizes partial AUC to preserve the accuracy of any given model. In summary, the contributions of this paper can be summarized as follows:
\begin{itemize}
	\item We challenge the commonly held belief that reducing popularity bias results in a trade-off with model performance. Through the analysis under general bias-variance decomposition framework, we demonstrate that bias reduction and model performance can be a win-win situation under certain conditions.
	
	\item We demonstrate how bias is related to the empirical estimates of false positive rate (FPR) and false negative rate (FNR). Additionally, by refining partitioning based on the popularity of items, we provide a comprehensive metric for measuring popularity bias, including overestimated hot items, overestimated cold items, underestimated hot items, and underestimated cold items.
	
	\item We provide the optimal negative sampling rule in terms of partial AUC maximization. It is achieved by modifying sample information and prior information that targeting on model-dependent and model-independent bias, which offers a flexible and principled way to reduce popularity bias in a lightweight manner without compromising the recommendation accuracy of the given recommendation model.
\end{itemize}

\section{Understanding Popularity Bias}
Denote an user item pair $(u,i)$ as a sample $\mathbf{x}$, where $u\in \mathcal{U}, i\in \mathcal{I}$. Let $\mathcal{X}= \{\mathbf{x}|u\in \mathcal{U}, i\in \mathcal{I}\} $ be the sample space indicating all the user item pairs and $\mathcal{Y} \subseteq \mathbb{R}$ be the output space indicating all the user preference. We assume the training data $\mathcal{D}=\{(x_i,y_i), i=1,2,...,n \}$ are i.i.d drawn from the unknown joint distribution $p(x,y)$ on $\mathcal{Z}=\mathcal{X}\times\mathcal{Y} $. A decision function $f:\mathcal{X} \rightarrow \mathbb{R}$ assigns a real valued preference level indicating the predicted preference value $f(\mathbf{x}) \in \mathbb{R}$. In this section, we first demonstrate how the model performance and bias (therefor popularity bias) are related. On the biases of which we introduce how to measure popularity bias. 

Considering the squared loss $(g(\mathbf{x}) - y)^2$ between the model prediction $g(\mathbf{x})$ and the ground truth label for a single data point $\mathbf{x}$. Since $g(\mathbf{x})$ is trained on training data set $\mathcal{D}$, we write it as $g(x;\mathcal{D}) $, which is a function of particular training set $\mathcal{D}$, indicating the impact of different training sets $\mathcal{D}$ on predictive values. 
Then, the expected squared loss over different data points $\mathbf{x}$ and training set $\mathcal{D}$ can be calculated as
\begin{eqnarray}
	\label{eq:sqarederror}
	\textsc{Mse}_{\mathcal{D}} =	\int_{\mathbf x\in \mathcal{X}}\int_{\mathcal{D}}(g(\mathbf{x};\mathcal{D}) - y)^2p(\mathbf{x},y)p(\mathcal{D})d\mathbf{x}d\mathcal{D},
\end{eqnarray}
which measures the model performance, $\textsc{Mse}_{\mathcal{D}} \rightarrow 0$ as $g(\mathbf{x};\mathcal{D}) = y$ almost surely for any $\mathbf{x}$ and $\mathcal{D}$. 

The popularity bias refers to the tendency of collaborative filtering models to excessively recommend popular items or overestimate users' preference for popular items, as observed in many datasets. Formally, it can be measured by the extent to which the average prediction over all datasets $\mathbb{E}_\mathcal{D}[g(\mathbf{x};\mathcal{D})]$ differs from the optimal prediction $y^*$.
\begin{eqnarray}
	\textsc{Bias}_\mathcal{D} = \int_{x\in \mathcal{X}}p(\mathbf{x},y)\int_{\mathcal{D}}\{\mathbb{E}_\mathcal{D}[g(\mathbf{x};\mathcal{D})] - y^*\}^2p(\mathcal{D})d\mathcal{D}d\mathbf{x}  \label{eq:bias}.
\end{eqnarray}
Mathematically, as described in the equation above, even after iterating through all possible training sets $\mathcal{D}$, the model's average prediction cannot approach the optimal prediction $y^*$. This represents a systemic error of the model, known as bias.

It is generally believed that reducing popularity bias would lead to a decrease in model performance. However, the following theorem provides a counterintuitive yet interesting conclusion: the relationship between model performance and systemic bias is not trade-off, but rather a positive correlation.
\begin{proposition}\label{theorem:bias_mse}
	Denote $y^* =\mathbb E_y [y|\mathbf{x}]$ as the Bayesian optimal prediction over all measurable functions. For  fixed complexity of the functional class $\mathcal{G}$ and joint distribution $p(x,y)$ on $\mathcal{Z}=\mathcal{X}\times\mathcal{Y}$, assuming the noise $y^*-y$ is zero-mean, then
	\begin{eqnarray}\label{eq:biaseq}
		\textsc{Bias}_\mathcal{D} &=& \textsc{Mse}_{\mathcal{D}}  - const 
	\end{eqnarray}
\end{proposition}
\begin{proof}
	See Appendix~\ref{appendix:bias_mse}.
\end{proof}

\subsection{Measure Popularity Bias}\label{sec:popmetric}
Note that the definition of bias, i.e., systematical error, in Eq~\eqref{eq:bias}, the expectation is taken over different $\mathcal{D}$, it is difficult to obtain the empirical estimate of Eq~\eqref{eq:bias} since only one $\mathcal{D}$ is available in practice. Fortunately, with the result of Proposition~\ref{theorem:bias_mse}, we can relate the bias with false positive rate(FPR), false negative rate(TPR), respectively.
\begin{proposition}\label{theorem:metric}
	Let $g(\mathbf{x}) \in \{0,1\}$ and $y \in \{0,1\}$, that is, simplifying implicit collaborative filtering to a classification problem. With the result of proposition~\ref{theorem:bias_mse}, we have
	\begin{eqnarray}
		bias \simeq \textsc{FPR}+\textsc{FNR}
	\end{eqnarray}
	\begin{proof}
		See Appendix~\ref{appendix:metric}.
	\end{proof}
\end{proposition}
\begin{figure}[h!]
	\centering
	\includegraphics[scale=0.5]{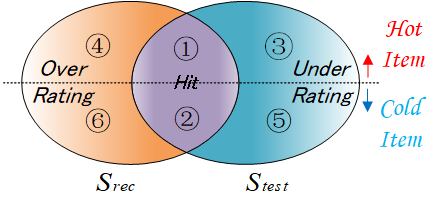}
	\caption{Interpretation of the four popularity bias metrics, which corresponds to over-recommended hot item rate (OHR=$\frac{\textcircled{4}}{\textcircled{4}+\textcircled{1}}$), over-recommended cold item rate (OCR=$\frac{\textcircled{6}}{\textcircled{6}+\textcircled{2}}$), under-recommended hot item rate (UHR=$\frac{\textcircled{3}}{\textcircled{3}+\textcircled{1}}$), and under-recommended cold item rate (UCR=$\frac{\textcircled{5}}{\textcircled{5}+\textcircled{2}}$).}
	\label{Fig:openball}
\end{figure}
Consequently, bias can be empirically estimated over samples $\mathbf{x}$ instead of $\mathcal{D}$. As the popularity of recommended (exposed) items $\mathcal{S}_\textsc{Rec}$ is our primary concern, we propose to  refined partition the recommended (exposed) items $\mathcal{S}_\textsc{Rec}$ based on their popularity, and use following four metrics: \textbf{o}ver-recommended \textbf{h}ot item \textbf{r}ate, \textbf{u}nder-recommended \textbf{h}ot item \textbf{r}ate, \textbf{o}ver-recommended \textbf{c}old item \textbf{r}ate,
\textbf{u}nder-recommended \textbf{c}old item \textbf{r}ate, to evaluate the popularity bias
\begin{eqnarray}
	\textsc{OHR}= \frac{\{\mathcal{S}_\textsc{Rec}-\mathcal{S}_\textsc{Test}\} \cap \mathcal{X}_\textsc{Hot} }{|\mathcal{S}_\textsc{Rec} \cap \mathcal{X}_\textsc{Hot}|} \label{eq:ORH}~~~
	\textsc{UHR} = \frac{\{\mathcal{S}_\textsc{Test}-\mathcal{S}_\textsc{Rec}\} \cap \mathcal{X}_\textsc{Hot} }{|\mathcal{S}_\textsc{Rec} \cap \mathcal{X}_\textsc{Hot}|} \label{eq:URH}\\
	\textsc{OCR} =\frac{\{\mathcal{S}_\textsc{Rec}-\mathcal{S}_\textsc{Test}\} \cap \mathcal{X}_\textsc{Cold} }{|\mathcal{S}_\textsc{Rec} \cap \mathcal{X}_\textsc{Cold}|} \label{eq:ORC}~~~
	\textsc{UCR} =\frac{\{\mathcal{S}_\textsc{Test}-\mathcal{S}_\textsc{Rec}\} \cap \mathcal{X}_\textsc{Cold} }{|\mathcal{S}_\textsc{Rec} \cap \mathcal{X}_\textsc{Cold}|} \label{eq:URC}
\end{eqnarray}
$\mathcal{X}_\textsc{Hot} = \{\mathbf{x}|(u,i) \in \mathcal{X}, pop_i \geq \tau  \} $ is the set of samples with item popularity greater than a predefined threshold $\tau$ according to specific scenarios,  $\mathcal{X}_\textsc{Cold} = \mathcal{X} - \mathcal{X}_\textsc{Hot}$. $\mathcal{S}_\textsc{Rec}-\mathcal{S}_\textsc{Test}$ is the set of items recommended to the user that they do not like, i.e., over-recommended false positive items. $\mathcal{S}_\textsc{Test}-\mathcal{S}_\textsc{Rec}$ is the set of items that the user likes but were not recommended to them, i.e., under-recommended false negative items. The fundamental philosophy behind the popularity bias metrics is not to discourage recommending popular items, but to avoid introducing excessive OHR caused by incorrectly recommending hot items. The key is to accurately estimate user preferences towards popular items, in order to match their true preferences and avoid overestimation or underestimation.

\section{Method}
Ranking-oriented or contrastive-based learning algorithms require positive-negative sample pairs to calculate the loss. Different negative sampling strategies determine which samples are pushed apart (and thus not exposed to the user). Therefore, we can intervene in the model training using negative sampling strategies to achieve end-to-end training and reduce popularity bias. 

\subsection{Empirical partial AUC Maximization}
For implicit collaborative filtering tasks, the top-k evaluation of personalized ranking task can be formalized as a partial AUC optimization problem~\cite{narasimhan2013structural,iwata2020semi,yang2021all,yang2022optimizing,zhu2022auc}, aiming at achieving a high
true positive rate  with a specific false positive rate range $[0,\gamma]$~\cite{dodd2003partial,yang2022auc}:
\begin{eqnarray}\label{eq:auc}
	pAUC(0,\gamma) = \frac{1}{|\mathcal{D}^+||\mathcal{N}|}\sum_{\mathbf{x}^+ \in \mathcal{D}^+}\sum_{\mathbf{x}^- \in \mathcal{N}} \mathbb{I}(g(\mathbf{x}^+) - g(\mathbf{x}^-))
\end{eqnarray}
where $\mathcal{N}\subseteq \mathcal{D}^-$ denotes a subset of negative samples ranked in the top $\lfloor \gamma\cdot|\mathcal{D}^-| \rfloor$-th position among negatives in a descending order of scores $g(\mathbf x)$. When taking $\gamma=1$, $pAUC(0,1)$ is standard AUC with false positive rate ranged in $[0,1]$.

Given each positive instance $\mathbf{x}^+$, a randomly sampled negative instance $\mathbf{x}$ was directly assigned with the negative gradient, resulting the \textit{minus effect} of $g(\mathbf{x})$. The following theorem states that negative sampling problem  can be reformulated as an AUC maximization problem that choosing unlabeled instance $\mathbf x$ to maximize partial AUC metric.
\begin{theorem}[Partial AUC Maximization]\label{theorem:auc}
	Denote the posterior probability of $\mathbf{x}\in \textsc{Tn}$ as $ \mathbb{P}(\textsc{Tn}|\mathbf{x})$. Given training set $\mathcal{D}= \mathcal{D}^+ \cup \mathcal{D}^-$, for each unit decrement $dg(\mathbf{x})=-1$, the following negative sampling strategy  
	\begin{eqnarray}
		\mathbf{x} 
		&=& \arg \max _\mathbf{x} \triangle_{\mathbf{x}^+}\cdot \mathbb{P}(\textsc{Tn}|\mathbf{x}) -\triangle_{\mathbf{x}^-}\cdot \mathbb{P}(\textsc{Fn}|\mathbf{x}) \label{eq:rule}
	\end{eqnarray}
	is the optimal sampling rule that maximizes the partial AUC metric over training set $\mathcal{D}$, where $\triangle_{\mathbf{x}^+}=\sum_{\mathbf{x}^+ \in \mathcal{D}^+}[1- \sigma(g(\mathbf{x}^+) - g(\mathbf{x}))]$ and $\triangle_{\mathbf{x}^-} = \sum_{\mathbf{x}^-\in \mathcal{N}} [1- \sigma(g(\mathbf{x}) - g(\mathbf{x}^-))]$, and $ \mathcal{N}\subseteq \mathcal{D}^-$ is the set of top-$\lfloor \gamma\cdot|\mathcal{D}^-| \rfloor$ ranked negative instances.
	\begin{proof}
		See Appendix~\ref{Appendix:auc}.
	\end{proof}
\end{theorem}


\subsection{Posterior Probability Estimation}
With the negative sampling rule of partial AUC maximization in Eq~\eqref{eq:rule}, we will present the posterior estimation of $\mathbf{x}\in \textsc{Tn}$. For a decision function $g: \mathcal{X} \rightarrow \mathbb{R} $, we denote the distribution of predicted preference values as $\phi(\hat{x})$, where $\hat{x}$ is short for $ g(x)$. To do this, we first introduce the technique of order statistics as a powerful nonparametric method in many areas of statistical theory and practice.

Consider $n$ random variables from $\phi$ arranged in the ascending order according to their realizations. We write them as $X_{(1)} \le X_{(2)} \le \cdots \le X_{(n)} $, and $X_{(k)}$ is called the $k$-th $(k=1,\cdots,n)$ order statistics~\cite{David:2004}. The \textit{probability density function} (PDF) of $X_{(k)}$ is given by:
\[
\phi_{(k)}(\hat{x}) = \frac{n!}{(k-1)!(n-k)!}\Phi^{k-1}(\hat{x}) \phi(\hat{x}) [1-\Phi(\hat{x})]^{n-k}
\]
Since only two populations are contained in the unlabeled dataset: true negative and false negative, so we take $n=2$
\begin{eqnarray}
	\phi_{(1)}(\hat{x}) = 2\phi(\hat{x}) [1-\Phi(\hat{x})] \label{eq:order1},~~~
	\phi_{(2)}(\hat{x}) = 2\phi(\hat{x}) \Phi(\hat{x}) \label{eq:order2}
\end{eqnarray}
where $\phi_{(1)}$ is the distribution of $X_{(1)}$ and $\phi_{(2)}$ is the distribution of $X_{(2)}$.  

Let $\alpha$ denote the probability that the decision function assigns a higher score to a positive example than a negative example, i.e., $\hat{x}^- \leq \hat{x}^+$ holds with probability $\alpha$, and $\hat{x}^+ \leq \hat{x}^-$ holds with probability $1-\alpha$. In terms of order statistics, the negative example $\hat{x}^-$ follows the distribution of the population $\phi_{(1)}$ with probability $\alpha$, and the distribution of population $\phi_{(2)}$ with probability $1-\alpha$. Thus, the distribution of true negative examples can be expressed as follows:
\begin{eqnarray}
	\phi_{\textsc{Tn}}(\hat{x}) &=& \alpha \phi_{(1)}(\hat{x}) + (1-\alpha)\phi_{(2)}(\hat{x}) \label{eq:phitn} \\
	&=& 2\alpha\phi(\hat{x}) [1-\Phi(\hat{x})]+ 2(1-\alpha)\phi(\hat{x}) \Phi(\hat{x}) \nonumber \\
\end{eqnarray}
Similarly, $\phi_{\textsc{Fn}}$ is the component of $X_{(2)}$ and $X_{(1)}$ with mixture coefficient $\alpha$:
\begin{eqnarray}
	\phi_{\textsc{Fn}}(\hat{x}) =  2\alpha \phi(\hat{x}) \Phi(\hat{x}) +2(1-\alpha)\phi(\hat{x}) [1-\Phi(\hat{x})] \label{eq:phifn}
\end{eqnarray}
So the posterior probability of $\mathbf{x} \in \textsc{Tn}$ can be calculated as 
\begin{eqnarray}\label{eq:posterior}
	\mathbb{P}(\textsc{Tn}|\mathbf{x}) &=& \frac{\phi_{\textsc{Tn}}(\hat{x}) \tau^-}{\phi_{\textsc{Tn}}(\hat{x}) \tau^- + \phi_{\textsc{Fn}}(\hat{x}) \tau^+} \nonumber \\
	&=&  \frac{\alpha\tau^- + (1-2\alpha) \Phi(\hat{x})\tau^-}
	{\alpha\tau^- + (1- \alpha)\tau^+  + (1-2\alpha)\Phi(\hat{x})(\tau^- -\tau^+)} \label{Eq:poster}
\end{eqnarray}
Due to the fractional form of conditional probability, $\phi(\hat x)$ is eliminated and only the commutative distribution function $\Phi(\hat{x})=\int_{-\infty} ^{\hat{x}} \phi(t)dt$  remains, which can be approximated by following empirical distribution function:
\begin{eqnarray}
	\Phi_n (\hat{x}) = \frac{1}{n} \sum_{i=1}^{n}\mathbb{I}_{|X_i \leq \hat{x}|}\label{eq:empiricalcdf}
\end{eqnarray}
Here, $\tau^- = 1-\tau^+$ is the prior probability of $\mathbb P (\mathbf{x} \in \textsc{Tn})$ that user do not prefer the item. This allows for penalizing popular items accordingly
\begin{eqnarray}
	\tau^- &=&\mathbb P (\mathbf{x} \in \textsc{Tn})\propto pop_i^\beta \label{eq:prior}
\end{eqnarray}

\subsection{Interpretation}
From a statistical perspective, $\Phi_n (\hat{x})\in [0,1]$ is the \textit{sample information} that is dependent on the model and reflects the joint probability of the current model classifying $\mathbf{x}$ as a positive example. For a sample with $\Phi_n (\hat{x})=1$, i.e., a top-ranked sample, the model assigns a probability of 1 to it being a positive example. On the other hand, $pop_i$ is the \textit{prior information} that is independent of the model but related to the data (i.e., item popularity). The model-dependent bias is manifested in the sample information, while the model-independent bias is manifested in the prior information. 

\textbf{Modifying sample information}: $\alpha\in[0.5,1]$ modifies the sample information to reduce the model-dependent bias. Specifically, $\alpha$ is the probability  of scoring positive instance higher than a negative instance, which admits  empirical macro-AUC estimate of the current model. For a worst model that  that makes random guesses, $\alpha=0.5$. It controls the \textit{confidence level} of the sample information and thus the degree of interference of sample information on negative sampling. When $\alpha$ approaches 0.5, i.e., indicating a very low confidence level of the sample information, the posterior probability degenerates into the prior probability	$\lim_{\alpha \rightarrow 0.5}\mathbb{P}(\textsc{Tn}|\mathbf{x}) = \tau^-$,
and only prior information is involved in negative sampling as random negative sampling (RNS) or popularity-based negative sampling (PNS). In contrast, when $\alpha$ approaches 1, indicating higher confidence of sample information, items with higher scores or higher ranking positions have higher posterior probabilities and are less likely to be sampled, thus less likely to be penalized as
$\lim_{\alpha \rightarrow 1.0,\Phi(\hat{x}) \rightarrow 1.0}\mathbb{P}(\textsc{Tn}|\mathbf{x}) =  0$.
$\alpha$ is can be used to control biases that originated from the model. For a model more prone to popularity bias, i.e., with low confidence in the sample information, a relative lower value of $\alpha$ should be chosen, thereby reduces the degree of interference of sample information in model training.

\textbf{Modifying prior information}: $\beta\in[0,+\infty]$ modifies the prior information to reduce the model-independent bias. Specifically, $\beta$ is the concentration parameter that controls the probability of popular items being true negatives. When $\beta$ approaches 0, the prior probability approaches uniform distribution. When $\beta$ approaches $\infty$, the prior probability of top popular item approaches 1. With fixed sample information, larger $\beta$ making popularity items more likely to be penalized. 

$\beta$ can be used to control biases that originate from outside the model. For datasets or scenarios that are more prone to popularity bias, a larger value of $\beta$ should be set to penalize popular items that are more likely to be over recommended.

\textbf{Maintaining recommendation accuracy}: $\gamma\in[0,1]$ controls the range of false positive rates that we aim to achieve a high true positive rate within. For top-k evaluation tasks, a smaller value of $\gamma$ should be chosen as the value of k decreases.

\subsection{Implementation}
Since $\triangle_{\mathbf{x}^+}$ and $\triangle_{\mathbf{x}^-}$ involves the summation over all samples, which introduces additional computational cost. To obtain a more practical form, we randomly draw N positive samples $\{ \mathbf{x}^+ _i\}_{i=1}^N$, and  N negative samples $\{ \mathbf{x}^- _j\}_{j=1}^N$ to  empirically estimate $\sum_\mathbf{x}^+$ and $\sum_\mathbf{x}^-$, respectively
\begin{eqnarray}
	\triangle_{\mathbf{x}^+} = |\mathcal{I}_u^+|\cdot \bar{\textsc{Info}}(\mathbf{x}^+) \label{eq:sum_plus}, ~~~
	\triangle_{\mathbf{x}^-} = \gamma|\mathcal{I}_u^-|\cdot \bar{\textsc{Info}}(\mathbf{x}^-)\label{eq:sum_minus}
\end{eqnarray}
where $
\bar{\textsc{Info}}(\mathbf{x}^+) \simeq\frac{1}{N}\sum_{i=1}^{N}[1- \sigma(g(\mathbf{x}^+_i) - g(\mathbf{x}))]  \label{eq:info+}$ and $
\bar{\textsc{Info}}(\mathbf{x}^-) \simeq \frac{1}{N}\sum_{i=1}^{N}[1- \sigma(g(\mathbf{x}) - g(\mathbf{x}^-_j))] \label{eq:info-}
$. Thus calculating $\triangle_{\mathbf{x}^+}$ and $\triangle_{\mathbf{x}^-}$ only takes $\mathcal{O}(N)$, where $N$ is a small constant. Given a small candidate set of negative samples $\mathcal{M}$, the algorithm for proposed negative sampling in terms of partial AUC maximization is presented in Algorithm~\ref{Alg:1}.
\begin{algorithm}[h]
	\caption{The proposed negative sampling algorithm}\label{Alg:1}
	\normalem
	\KwIn{Positive instance $\mathbf{x}^+$, candidate negative set $\mathcal{M}$, decision function $g(\cdot)$, $\alpha$, $\beta$, $\gamma$.}
	\KwOut{Negative instance}
	\For{candidate unlabeled sample $\mathbf{x}$ in candidate set $\mathcal{M}$} {
		~~Calculate empirical C.D.F by Eq~\eqref{eq:empiricalcdf}. \label{line2}\\
		Calculate prior probability by Eq~\eqref{eq:prior}.\label{line2.1}\\
		Calculate posterior estimation $\mathbb{P}(\textsc{Tn}|\mathbf{x})$ by Eq~\eqref{eq:posterior}. \label{line3}\\
		Calculate $\triangle_{\mathbf{x}^+},\triangle_{\mathbf{x}^-}$ by Eq~\eqref{eq:sum_plus}.\label{line4} \\
		Select negative instance $\mathbf{x}$ by Eq~\eqref{eq:negsampling}. \label{line5}\\}
	\KwResult{Negative instance $\mathbf{x}$.}
\end{algorithm}

\textbf{Complexity}: The computational complexity of Algorithm~\ref{Alg:1} mainly comes from line~\ref{line2}, which takes $\mathcal{O}(|\mathcal{I}_u^-|)$; Line~\ref{line3} takes $\mathcal{O}(1)$, line~\ref{line4} takes $\mathcal{O}(M)$, Line~\ref{line5} takes $\mathcal{O}(1)$. So the complexity of Algorithm~\ref{Alg:1} is $\mathcal{O}(M|\mathcal{M}||\mathcal{I}_u^-|)$, which is still linear with respect to number of un-interacted items $|\mathcal{I}_u^-|$ as random negative sampling.

\section{Experiments}
\subsection{Experiment Settings}
\subsubsection{Dataset}
We conduct experiments on three public datasets, including MovieLens-100k\footnote{~\url{https://grouplens.org/datasets/movielens}}, MovieLens-1M, and Yahoo!-R3.\footnote{~\url{http://webscope.sandbox.yahoo.com/catalog.php?datatype=r}.}  We  convert all the rating data~\cite{Steffen:2009:UAI} into implicit feedback and partition it into a training set, consisting of 80\% of the data, and a test set, consisting of the remaining 20\%. Table~\ref{Table:Dataset} of Appendix~\ref{appendix:data} summarizes the dataset statistics. We define the top 15\% of items ranked by interaction frequency as popular items.

\subsubsection{Baselines}
We compare proposed approach against two categories of methods. The first category is negative sampling algorithms, including  (a) fixed negative sampling distribution such as \textsf{RNS}~\cite{Steffen:2009:UAI,Xiangnan:2020:SIGIR,Weike:2013:IJCAI,Yu:2018:CIKM,Wang:2019:SIGIR,Xuejiao:2020:ASC,Liu:2017:SIGIR}, \textsf{PNS}~\cite{Mikolov:2013:NIPS,Chen:2017:KDD,Tang:2015:WWW}and FairStatic~\cite{Chen:2023:WWW}. (b) hard negative sampling with dynamic sampling distribution, including  \textsf{DNS}~\cite{Zhang:2013:SIGIR}.
\begin{itemize}
	\item[-]\textsf{RNS}:~\cite{Steffen:2009:UAI,Xiangnan:2020:SIGIR,Weike:2013:IJCAI,Yu:2018:CIKM,Wang:2019:SIGIR,Xuejiao:2020:ASC,Liu:2017:SIGIR}: (Random Negative Sampling) Uniformly sampling negatives.
	\item[-]\textsf{PNS}:~\cite{Mikolov:2013:NIPS,Chen:2017:KDD,Tang:2015:WWW}: (Popularity-based Negative Sampling) Adopting a fixed distribution proportional to an item interaction ratio, i.e., $\propto pop_j ^{0.75}$.
	\item[-]\textsf{DNS}:~\cite{Zhang:2013:SIGIR}: Over-sampling relative higher ranked hard negatives. Its sampling probability is a liner function to the relative ranking position.
	\item[-]\textsf{FairStatic}\cite{Chen:2023:WWW}: Sampling probabilities are calculated based on group popularity to achieve fair recommendations across different item categories.
\end{itemize}
The second category of baselines is popularity  debiasing  methods: including  (a)sample re-weighting method IPS-CN~\cite{Gruson:2019:WSDM}, (b) causal inference methods such as~\cite{Wei:2021:KDD} and CausE~\cite{Bonner:2018:RecSys} (c) loss function correcting methods such as BC\_LOSS~\cite{Zhang:2022:NIPS}, SAMREG~\cite{Boratto:2021:IP&MC}.
\begin{itemize}
	\item[-]\textsf{IPS-CN}:~\cite{Gruson:2019:WSDM}: Normalization or smoothing IPS penalty to attain a more stable output.
	\item[-]\textsf{MACR}:~\cite{Wei:2021:KDD}: Utilizing the counterfactual world in causal inference theory to decouple users, items, and recommendation results.
	\item[-]\textsf{CausE}:~\cite{Bonner:2018:RecSys}: Integrating recommendation behavior and user natural behavior into a fusion model based on causal inference.
	\item[-]\textsf{BC\_LOSS}:~\cite{Zhang:2022:NIPS}: Correcting contrastive loss by using bias-awareness angle. 
	\item[-]\textsf{SAMREG}:~\cite{Boratto:2021:IP&MC}: Correcting the loss function through a regularization term.
\end{itemize}
\subsubsection{Evaluation metrics}
We adopt top-k ranking metrics, including precision@k, recall@k, F1@k, and NDCG@k to evaluate whether top-k  recommended (exposed) items match users' preference. Given their widespread use, we omit their definitions here. Additionally, we use overestimation and underestimation rates associated with item popularity to evaluate the popularity bias. The definitions of these metrics are provided in Section\ref{sec:popmetric}. 

\subsubsection{Parameter Settings}
We employ the classic \textit{matrix factorization} (MF)~\cite{Koren:2009:Computer} and the recent \textit{light  graph convolution network} (LightGCN)~\cite{Xiangnan:2020:SIGIR} as encoder. The detailed parameter settings are presented in Table~\ref{table:Parameter settings} of Appendix~\ref{apendix:Parameter settings}. The codes is implemented with Pytorch~\footnote{Codes and data are released at:~\url{https://github.com/studentcej/AUC_NS}}. Computations were conducted on a personal computer with Windows 10 operating system, 3.99 GHz CPU,  RTX 3060Ti GPU, and 32 GB RAM.
\subsection{Experimental Results}
\subsubsection{Overall Performance}
Table~\ref{tab:performance} shows the overall ranking performance and popularity debiasing performance. The popularity debiasing task focuses on the over-recommended hot item rate (OHR). FairStatic~\cite{Chen:2023:WWW} and IPS-CN~\cite{Gruson:2019:WSDM} significantly reduce the OHR metric, but they also introduce too much bias of other kinds, which fails to maintain recommendation accuracy. When measuring the performance of debiasing methods, it is necessary to consider whether they can maintain accuracy. For example, uniformly recommending all items may lead to a decrease in the OHR, but it also results in the loss of recommending items that match users' preferences. 

Across two models and three datasets, our proposed method significantly reduces the over-recommended hot item rate (OHR) without compromising recommendation accuracy. These results suggest that reducing bias and improving recommendation accuracy can be a win-win scenario. Compared to negative sampling methods that only utilize prior information of item popularity (such as RNS, PNS) or sample information of ranking position (such as DNS), our proposed method combines item popularity and model prediction to formulate the posterior information for maximizing partial AUC, resulting in the best ranking performance. Furthermore, compared to popularity debiasing methods, our proposed method considers both model-dependent and model-independent bias, reduces the over-recommended popular items by appropriately penalizing them, and avoids introducing bias of under-recommended popular items.

\begin{table*}[!]
	\centering
	\caption{Overall performance of top-5 recommendation.}\label{tab:performance}
	\renewcommand{\arraystretch}{0.8}
	\resizebox{1\textwidth}{!}
	{
		\begin{threeparttable}
			\begin{tabular}{ccccccccccc}
				\toprule
				\multirow{3}{*}{Dataset} & \multirow{3}{*}{CF Model} & \multirow{3}{*}{Method} & \multicolumn{4}{c}{Ranking Perfomence} & \multicolumn{4}{c}{Popularity Debiasing Performence} \\
				\cmidrule(lr){4-7} \cmidrule(lr){8-11}          &       &       & Pre & Recall & F1 & NDCG & OHR & OCR & UCR & UHR \\
				\midrule
				\multirow{20}[4]{*}{MovieLens-100K} & \multirow{10}[2]{*}{MF} & RNS   &0.3676 & 0.1259 & 0.1614 & 0.3938 & 0.6324 & 0.0021 & 0.9289 & 0.8034 \\
				&       & FairStatic &0.3519 & 0.1341 & 0.1688 & 0.3817 & 0.5353 & 0.3732 & 0.9056 & 0.8117  \\
				&       & DNS   & 0.3843 & 0.131 & 0.1683 & 0.4086 & 0.6157 & 0.0021 & 0.9289 & 0.7939 \\
				&       & PNS   & \underline{0.3979} & 0.1338 & 0.1742 & \underline{0.4231} & 0.5939 & 0.1336 & 0.9225 & 0.7978 \\
				&       & IPS-CN &0.2891 & 0.0812 & 0.112 & 0.3036 & 0.5724 & 0.4948 & 0.8963 & 0.8966 \\
				&       & MACR  &0.3941 & 0.1354 & 0.1748 & 0.4178 & 0.5907 & 0.1371 & 0.9215 & 0.797\\
				&       & SAMREG &0.3919 & 0.1343 & 0.1736 & 0.4186 & 0.5874 & 0.19 & 0.9154 & 0.8042 \\
				&       & CausE &0.3947 & \underline{0.1357} & \underline{0.1753} & 0.4178 & 0.5943 & 0.1209 & 0.921 & 0.7953\\
				&       & BC\_LOSS &0.3849 & 0.1192 & 0.1584 & 0.3999 & 0.6021 & 0.1994 & 0.9133 & 0.8248  \\
				&       & \textbf{Proposed} &\textbf{0.4324} & \textbf{0.1462} &\textbf{0.1892}  & \textbf{0.4620} & 0.5671 & 0.0112 & 0.9270 & 0.7721   \\
				\cmidrule{2-11}          & \multirow{10}[2]{*}{LightGCN} & RNS   &0.3799 & 0.1302 & 0.1671 & 0.4028 & 0.6199 & 0.0053 & 0.9289 & 0.7966 \\
				&       & FairStatic & 0.3705 & 0.1328 & 0.1694 & 0.3986 & 0.5699 & 0.3213 & 0.9118 & 0.8103 \\
				&       & DNS   & 0.3917 & 0.1337 & 0.172 & 0.4166 & 0.6082 & 0.0064 & 0.9279 & 0.7903 \\
				&       & PNS   & 0.4098 & 0.1395 & 0.1802 & 0.4335 & 0.5813 & 0.1175 & 0.9218 & 0.7897  \\
				&       & IPS-CN &0.3379 & 0.1164 & 0.1499 & 0.3602 & 0.645 & 0.196 & 0.9185 & 0.8284 \\
				&       & MACR  &   0.4081 & 0.1408 & 0.1817 & 0.435 & 0.5747 & 0.1881 & 0.9151 & 0.7921\\
				&       & SAMREG & 0.4078 & 0.1383 & 0.1796 & 0.4357 & 0.5673 & 0.2208 & 0.9127 & 0.7992  \\
				&       & CausE & \underline{0.4142} & \underline{0.1441} & \underline{0.185}4 & \underline{0.4367} & 0.5671 & 0.1847 & 0.9148 & 0.7874  \\
				&       & BC\_LOSS & 0.3968 & 0.1318 & 0.1714 & 0.418 & 0.5999 & 0.058 & 0.9257 & 0.7926 \\
				&       & \textbf{Proposed} &\textbf{0.4380}  & \textbf{0.1484} &\textbf{0.1918}  & \textbf{0.4669} & 0.5614 & 0.0134 & 0.9265 & 0.7674 \\
				\midrule
				\multirow{20}[4]{*}{Yahoo!-R3} & \multirow{10}[2]{*}{MF} & RNS   &0.1334 & 0.0989 & 0.1088 & 0.1495 & 0.861 & 0.0771 & 0.9829 & 0.7232 \\
				&       & FairStatic &0.1348 & 0.1015 & 0.1111 & 0.1532 & 0.8565 & 0.2156 & 0.9811 & 0.7186 \\
				&       & DNS   &  \underline{0.139} & \underline{0.1037} & \underline{0.1139} & \underline{0.1558} & 0.8572 & 0.0733 & 0.9828 & 0.7109  \\
				&       & PNS   & 0.1237 & 0.0916 & 0.1009 & 0.1394 & 0.8626 & 0.2415 & 0.9806 & 0.7438  \\
				&       & IPS-CN & 0.0321 & 0.0208 & 0.0239 & 0.033 & 0.3918 & 0.9755 & 0.9665 & 0.9195\\
				&       & MACR  &0.1335 & 0.0991 & 0.1089 & 0.1509 & 0.8653 & 0.0506 & 0.9838 & 0.7199  \\
				&       & SAMREG & 0.1257 & 0.0919 & 0.1017 & 0.1375 & 0.8655 & 0.1666 & 0.9807 & 0.7419  \\
				&       & CausE & 0.1224 & 0.0904 & 0.0996 & 0.1355 & 0.8727 & 0.0819 & 0.9846 & 0.7389 \\
				&       & BC\_LOSS & 0.1209 & 0.0863 & 0.0964 & 0.1288 & 0.8644 & 0.222 & 0.9777 & 0.7602  \\
				&       & \textbf{Proposed} & \textbf{0.1397}  & \textbf{0.1045} & \textbf{0.114}6 & \textbf{0.1583} & 0.8600 & 0.0317 & 0.9848 & 0.7052  \\
				\cmidrule{2-11}          & \multirow{10}[2]{*}{LightGCN} & RNS   &0.1436 & 0.1059 & 0.1168 & 0.1605 & 0.8394 & 0.1438 & 0.9792 & 0.7157   \\
				&       & FairStatic &0.1384 & 0.1037 & 0.1138 & 0.1572 & 0.8404 & 0.3723 & 0.9763 & 0.7229\\
				&       & DNS   & \underline{0.1481} & \underline{0.1096} & \underline{0.1207} & \underline{0.1666} & 0.8395 & 0.1337 & 0.979 & 0.7067 \\
				&       & PNS   &0.1282 & 0.0946 & 0.1044 & 0.145 & 0.8451 & 0.2994 & 0.9762 & 0.7444 \\
				&       & IPS-CN &0.1111 & 0.0819 & 0.0904 & 0.1262 & 0.8747 & 0.1223 & 0.9823 & 0.7602  \\
				&       & MACR  & 0.1023 & 0.0758 & 0.0835 & 0.1149 & 0.7876 & 0.1961 & 0.9848 & 0.7746  \\
				&       & SAMREG &0.1406 & 0.103 & 0.1139 & 0.1561 & 0.8414 & 0.1854 & 0.9779 & 0.7223   \\
				&       & CausE &0.1435 & 0.1067 & 0.1173 & 0.1624 & 0.8508 & 0.086 & 0.9815 & 0.7067 \\
				&       & BC\_LOSS &0.1336 & 0.0998 & 0.1094 & 0.1499 & 0.8629 & 0.0421 & 0.9846 & 0.7179  \\
				&       & \textbf{Proposed} &\textbf{0.1537} & \textbf{0.1144} & \textbf{0.1257} & \textbf{0.1763} & 0.8427 & 0.0940 & 0.9820 & 0.6906  \\
				\midrule
				\multicolumn{1}{c}{\multirow{20}[4]{*}{MovieLens-1M}} & \multirow{10}[2]{*}{MF} & RNS   & 0.3675 & 0.0774 & 0.1117 & 0.3884 & 0.6314 & 0.0118 & 0.9497 & 0.8785  \\
				&       & FairStatic &0.3543 & 0.09 & 0.1256 & 0.3781 & 0.5436 & 0.3945 & 0.9274 & 0.8786 \\
				&       & DNS   &0.3873 & 0.0848 & 0.1214 & 0.4081 & 0.6113 & 0.0144 & 0.9493 & 0.867 \\
				&       & PNS   &\underline{0.3931} & \underline{0.095} & \underline{0.1334} & \underline{0.4129} & 0.5949 & 0.1708 & 0.9398 & 0.8635  \\
				&       & IPS-CN &0.2437 & 0.0516 & 0.0738 & 0.2456 & 0.6334 & 0.3839 & 0.936 & 0.9358  \\
				&       & MACR  & 0.3884 & 0.0843 & 0.1211 & 0.4079 & 0.6114 & 0.0194 & 0.9492 & 0.8677   \\
				&       & SAMREG &0.3826 & 0.085 & 0.1214 & 0.4019 & 0.6154 & 0.0275 & 0.9484 & 0.8686  \\
				&       & CausE &0.3774 & 0.0812 & 0.1167 & 0.3986 & 0.6216 & 0.0131 & 0.9497 & 0.8731  \\
				&       & BC\_LOSS &0.349 & 0.0749 & 0.108 & 0.3642 & 0.6502 & 0.0305 & 0.9489 & 0.8817    \\
				&       & \textbf{Proposed} & \textbf{0.4218} & \textbf{0.0989} & \textbf{0.1398} & \textbf{0.4429} & 0.5786 & 0.0214 & 0.9484 & 0.8460  \\
				\cmidrule{2-11}          & \multirow{10}[2]{*}{LightGCN} & RNS   &0.3812 & 0.0864 & 0.1227 & 0.3985 & 0.6065 & 0.1497 & 0.9406 & 0.8737  \\
				&       & FairStatic & 0.2652 & 0.0656 & 0.0913 & 0.2803 & 0.3991 & 0.7177 & 0.8995 & 0.9365  \\
				&       & DNS   &0.3875 & 0.0883 & 0.1252 & 0.4064 & 0.6022 & 0.1488 & 0.94 & 0.8713 \\
				&       & PNS   & 0.3308 & 0.0757 & 0.1069 & 0.3477 & 0.6025 & 0.3729 & 0.9297 & 0.9019 \\
				&       & IPS-CN &  0.294 & 0.0631 & 0.0908 & 0.3128 & 0.6492 & 0.2794 & 0.9402 & 0.9124  \\
				&       & MACR  & 0.3871 & 0.0861 & 0.1229 & 0.4075 & 0.6117 & 0.0358 & 0.9484 & 0.8668 \\
				&       & SAMREG &0.3803 & 0.0871 & 0.1234 & 0.397 & 0.5981 & 0.2478 & 0.9321 & 0.8797 \\
				&       & CausE & \underline{0.3935} & \underline{0.09}2 & \underline{0.1297} & \underline{0.4123} & 0.5925 & 0.1825 & 0.9347 & 0.8688\\
				&       & BC\_LOSS &0.2126 & 0.041 & 0.0593 & 0.2222 & 0.7865 & 0.001 & 0.951 & 0.937 \\
				&       & \textbf{Proposed} &\textbf{0.4190} &\textbf{0.0991}  & \textbf{0.1395} & \textbf{0.4400} & 0.5772 & 0.1118 & 0.9404 & 0.8522   \\
				\bottomrule
			\end{tabular}
			\begin{tablenotes}    
				\footnotesize               
				\item[1] Improvements of precision, recall, F1, ndcg are significant (p $\leq$ 0.01) , as validated by student's t-test(expect yahoo-MF).
			\end{tablenotes}           
		\end{threeparttable}
	}
\end{table*}%

\subsubsection{Hyper-parameter Analysis}
\begin{table}[!]
	\centering
	\caption{Impact of $\alpha$.}
	\resizebox{0.8\textwidth}{!}{
		\begin{tabular}{lccccccccccc}
			\toprule
			$\alpha$     & 0.5   & 0.55  & 0.6   & 0.65  & 0.7   & 0.75  & 0.8   & 0.85  & 0.9   & 0.95  & 1 \\
			\midrule
			Precision   & 0.4263 & 0.4267 & 0.4308 & 0.4267 & 0.4297 & \textbf{0.4324} & 0.4295 & 0.4305 & 0.4297 & 0.4191 & 0.3941 \\
			OHR & 0.5723 & 0.5717 & 0.5678 & 0.5725 & 0.5691 & \textbf{0.5671} & 0.5697 & 0.5695 & 0.5692 & 0.5793 & 0.6103 \\
			\bottomrule
		\end{tabular}%
		\label{tab:alpha}%
	}
\end{table}%
The impact of $\alpha$ is presented in Table~\ref{tab:alpha}. $\alpha$ corresponds to macro-AUC of encoder that accounts for model-dependent bias correction. As $\alpha$ increases, the top-k ranking performance exhibit an inverted U-shaped curve, with an initial increase followed by a decline. This trend is due to the fact that $\alpha$ corresponds to the macro AUC empirical estimate of the encoder, and inappropriate $\alpha$ resulting an biased posterior estimation $\mathbb{P}(\textsc{Tn}|\mathbf{x})$, which decrease ranking performance. In terms of the popularity debiasing metric, the over-recommended hot item rate(OHR) remains relatively stable when $\alpha$ values are below the optimal level because it isolates the bias from the model. However, when $\alpha$ values are above the optimal level, the sample information is assigned an excessively high confidence. Given that the collaborative filtering model originally tends to overestimate user preference for popular items, the too large value of $\alpha$ further reduces the probability of sampling these true negative examples, since $\mathbb{P}(\textsc{Tn}|\mathbf{x})$ is a decreasing function of $\alpha$. At this point, the too large value of $\alpha$ not only fails to correct model bias but also further amplifies it, leading to a rapid increase in the overestimation rate of popular items. In summary, for models that are more prone to popularity bias, a relatively lower value of $\alpha$ should be selected.

\begin{table}[!]
	\centering
	\caption{The impact of $\beta$.}
	\resizebox{0.8\textwidth}{!}{
		\begin{tabular}{lcccccccc}
			\toprule
			$\beta$     & 0.1   & 0.05  & 0.025 & 0.01  & 0.0075 & 0.005 & 0.001 & 0.0005 \\
			\midrule
			Precision   & 0.3983 & 0.4197 & 0.4299 & \textbf{0.4324} & 0.4299 & 0.4282 & 0.4261 & 0.4233 \\
			OHR & 0.5871 & 0.5734 & 0.5679 & \textbf{0.5671} & 0.5686 & 0.5718 & 0.5721 & 0.5754 \\
			Sampled popular item rate    & 0.3604 & 0.3241 & 0.2926 & 0.2709 & 0.2674 & 0.2636 & 0.2584 & 0.2577 \\
			Sampled false negative rate   & 0.0701 & 0.0637& 0.0581 & 0.0539 & 0.0533 & 0.0526 & 0.0519 & 0.0517 \\
			\bottomrule
		\end{tabular}%
		\label{tab:beta}%
	}
\end{table}%
The impact of $\beta$ is presented in Table~\ref{tab:beta}. $\beta$ is a parameter used to correct model-independent bias (bias from dataset). As $\beta$ increases, negative sampling is dominated by a higher prior probability to increase the sampling of popular items and penalize them, thereby correcting the popularity bias from the dataset. With an increase in $\beta$, the top-k ranking metrics initially increases and then decreases, while the popularity debiasing metric (OHR) exhibits the opposite trend. This is because the sampled popular item rate  increases with $\beta$, which increases the penalty on popular items and corrects the bias from the dataset. However, when excessive popular items are sampled, the sampled false negative rate also increases, which means that popular items are excessively penalized, leading to a decrease in recommendation performance. Additionally, an excessively large value of $\beta$ can cause the prior probability to be overly concentrated on the most popular items, resulting in a loss of correction ability for dataset bias. Therefore, for datasets with severe popularity bias, a proper value of $\beta$ should be selected.

\begin{table}[!]
	\centering
	\caption{The impact of $\gamma$}
	\resizebox{0.8\textwidth}{!}{
		\begin{tabular}{cccccccccc}
			\toprule
			$\gamma$    & 2e-03  & 4e-03 & 6e-03& 8e-03 & 10e-03 & 100e-03 & 200e-03 & 300e-03 & 400e-03 \\
			\midrule
			Pre@5 & 0.4293 & 0.4276 & \textbf{0.4324} & 0.4305 & 0.4305 & 0.4246 & 0.4231 & 0.4129 & 0.407 \\
			Pre@10  & 0.3581  & 0.3599  & 0.3584 & 0.3594  & 0.3589 & \textbf{0.3601} & 0.3543 & 0.3504 & 0.3454 \\
			\bottomrule
		\end{tabular}%
		\label{tab:gamma}%
	}
\end{table}%
The impact of $\gamma$ is presented in Table~\ref{tab:gamma}. $\gamma\in[0,1]$ controls the range of false positive rates that we aim to achieve a high true positive rate within, to maintain the top-k ranking performance. The best top-5 ranking performance is achieved at $\gamma=6e-3$, while the best top-10 ranking performance is achieved at $\gamma=100e-3$, since higher false positive rates are tolerable as k increases. Moreover, as Eq.\eqref{eq:sum_minus} is a biased estimate due to the fact that negative examples are unlabeled (false negatives main contained), a appropriately selected $\gamma$ can compensate for the impact of false negative examples when estimating Eq.\eqref{eq:sum_minus}.

\section{Related Work}
We summarize related work from the following two aspects, including popularity debiasing methods and negative sampling methods.

The approach for reducing popularity bias can be categorized into three strands for. The first type is \textit{post-ranking methods}~\cite{Steck:2018:RecSys,Abdollahpouri:2019:arxiv,Steck:2019:arxiv,Zhu:2021:KDD,Zhu:2021:WSDM,Liu:2017:SIGIR}, which modify the rating matrix or recommendation list manually after model training. For example, FPC~\cite{Zhu:2021:KDD} rescales predicted scores to eliminate popularity bias, while Liu et al.~\cite{Liu:2017:SIGIR} increase recommendation diversity by randomly shuffling the list through random walking. The second type is \textit{empirical risk rewriting methods}\cite{Schnabel:2016:PMLR,Bottou:2013:JMLR,Gruson:2019:WSDM,Joachims:2017:IJCAI,Yang:2018:RecSys,Saito:2020:WSDM}, which re-weight samples targeting on popularity bias from dataset, with the most classic algorithm being IPS\cite{Schnabel:2016:PMLR}. However, the estimation of inverse propensity score is subject to limitations, and many studies have proposed more sophisticated methods for estimating them~\cite{Bottou:2013:JMLR,Gruson:2019:WSDM}. The third type is\textit{ causal inference methods}\cite{Bonner:2018:RecSys,Zhang:2021:SIGIR,Wang:2021:SIGDKK,Wei:2021:SIGDKK,Zheng:2021:WWW,Liu:2020:SIGIR,Zhan:2022:SIGKDD,Xu:2021:arXiv}, which utilize causal inference theory to decouple the undesired causal relationship in the recommendation model.
By using the Front-Door\cite{Xu:2021:arXiv}/Back-Door~\cite{Zhan:2022:SIGKDD}/intervention criteria~\cite{Zhang:2021:SIGIR} in causal inference to decouple the effect of item popularity and exposure. Similarly, MACR~\cite{Wei:2021:KDD} has decoupled the direct connection between item and user popularity and recommendation results to eliminate popularity bias.

From a statistical perspective, negative sampling methods can be classified into \textit{static sampling distribution} and \textit{dynamic sampling distribution} based on whether the sampling distribution for negative samples is dependent on model predictions. For the first category of static sampling distribution, the basic idea is to relate static prior information to the sampling distribution, and adopt a fixed distribution for negative sampling during the training process. The most widely used method is random negative sampling (RNS)\cite{Steffen:2009:UAI,Xuejiao:2020:ASC,Yu:2018:CIKM,Xiangnan:2020:SIGIR}, which uniformly samples negatives from un-interacted instances. Some studies have proposed setting the sampling probability of a negative instance according to its popularity (interaction frequency), known as popularity-biased negative sampling (PNS)\cite{Mikolov:2013:NIPS,Chen:2017:KDD, Tang:2015:WWW,Mihajlo:2015:SIGIR}. The second category is \textit{dynamic negative sampling}, which adopts an adaptive sampling distribution targeting hard negative instances. The basic idea is that negative instances similar to positive instances in the embedding space\cite{Zhang:2019:ICDE,Sun:2018:IJCAI} benefit model training. Many hard negative sampling strategies have been proposed for personalized recommendation\cite{Steffen:2014:WSDM,Zhang:2013:SIGIR,Ding:2020:NIPS,Park:2019:WWW,Huang:2021:KDD,Ding:2019:IJCAI,Ding:2020:NIPS}. For example, \cite{Zhang:2013:SIGIR} propose to oversample higher scored items, and \cite{Steffen:2014:WSDM} proposed to oversample higher ranked negatives, since they are similar to positive instances. \cite{Wentao:2023:WWW} provides a theoretical analysis of the relationship between hard negative sampling and one-way partial AUC optimization.

\section{Conclusion}
In this paper, we utilize negative sampling to intervene in model training, maximizing partial AUC to maintain recommendation accuracy while penalizing popular items to reduce popularity bias. Specifically, it corrects sample information to alleviate model-dependent bias and corrects prior information to alleviate model-independent bias. This provides a flexible and unified method to reduce popularity bias without sacrificing recommendation accuracy. However, proposed method only maintains the true positive rate for hot items, yet a common problem faced is that collaborative filtering models tend to recommend fewer cold items. Although we can increase the recommendation of cold items by excessively penalizing popular items, additional bias of over-recommended cold items may introduced, which is left to be addressed in future research by enhancing recommendation models.

\normalem
\bibliography{main}
\bibliographystyle{plain}

\newpage
\appendix
\onecolumn

\section{Experimental details}
\subsection{Dataset Statistics}\label{appendix:data}
\begin{table}[!]
	\centering
	\small
	\caption{Dataset Statistics}\label{Table:Dataset}
	\begin{tabular}{lrrcrr}
		\toprule[1.2pt]
		~           & users   & items  & number of popular items & training set  &test set  \\ \cline{1-6}
		MovieLens-100k   &   943    &  1,682 &252  &    80k	   & 20k 	\\
		MovieLens-1M    &   6,040  &  3,952& 592  &   800k     & 200k  \\
		Yahoo!-R3       &   5,400  &  1,000 &150  &   146k      & 36k  \\
		\bottomrule[1.2pt]
	\end{tabular}
\end{table}
\subsection{Parameter settings}\label{apendix:Parameter settings}
\begin{table}[htbp]
	\caption{Parameter settings.}
	\centering
	\resizebox{0.8\textwidth}{!}{	
		\begin{tabular}{lccccccc}
			\toprule
			& Batch Size    & Learning Rate    & lr\_Decay Rate(epoch) & Regulation Constant    & $\alpha$     & $\beta$     & $\gamma$ \\
			\midrule
			100K-MF & 128   & 0.1   & 0.1[20,60,80] & 1.00E-04 & 0.75  & 0.01  & 0.006 \\
			100K-LightGCN & 128   & 0.1   & 0.1[20,60,80] & 1.00E-04 & 0.75  & 0.01  & 0.006 \\
			1M-MF & 128   & 5.00E-04 & 1     & 1.00E-05 & 0.75  & 0.01  & 0.006 \\
			1M-LightGCN & 1024  & 5.00E-04 & 1     & 0     & 0.75  & 0.01  & 0.006 \\
			yahoo-MF & 128   & 5.00E-04 & 1     & 0     & 0.65  & 0.001 & 0.004 \\
			yahoo-LightGCN & 128   & 5.00E-04 & 1     & 0     & 0.65  & 0.001 & 0.004 \\
			\bottomrule
		\end{tabular}%
	}
	\label{table:Parameter settings}%
\end{table}%
\begin{table}[!]
	\caption{The performance of MSE, FPR, and FNR}\label{tab:msetpr}
	\centering
	\resizebox{0.8\textwidth}{!}{
		\begin{tabular}{ccccccccccc}
			\toprule
			&       &\multicolumn{3}{c}{MovieLens100K} &\multicolumn{3}{c}{MovieLens1M}&\multicolumn{3}{c}{Yahoo!-R3} \\ \cline{3-5} \cline{6-8} \cline{9-11}  
			\multirow{3}{*}{CF Model} & \multirow{3}{*}{Method} & \multirow{3}{*}{MSE} & \multicolumn{2}{c}{BIAS} & \multirow{3}{*}{MSE} & \multicolumn{2}{c}{BIAS}& \multirow{3}{*}{MSE} & \multicolumn{2}{c}{BIAS}\\
			\cmidrule{4-5} \cmidrule{7-8}     \cmidrule{10-11}            &       &       & FPR   & FNR&       & FPR   & FNR&       & FPR   & FNR \\
			
			\midrule
			\multirow{10}[2]{*}{MF} & RNS & 0.2361 & 0.6324 & 0.8741 &0.2495 & 0.6325 & 0.9226&0.2474 & 0.8666 & 0.9011 \\ 
			& FairStatic & 0.236 & 0.6481 & 0.8659&0.2482 & 0.6457 & 0.91&0.253 & 0.8652 & 0.8985  \\ 
			&DNS & 0.231 & 0.6157 & 0.869&0.2451 & 0.6127 & 0.9152& \underline{0.2398} & \underline{0.861} & \underline{0.8963} \\ 
			&PNS & 0.2339 & 0.6021 & 0.8662&0.2438 & \underline{0.6069} & \underline{0.905}& 0.2584 & 0.8762 & 0.9083  \\ 
			&IPS-CN & 0.2474 & 0.7101 & 0.9153&0.2655 & 0.7563 & 0.9501&0.2723 & 0.9679 & 0.9792  \\ 
			&MACR & 0.2374 & 0.6059 & 0.8635&0.2439 & 0.6327 & 0.9188&0.2546 & 0.8665 & 0.9009  \\ 
			&SAMREG & 0.2483 & 0.6081 & 0.8657&0.2572 & 0.6323 & 0.9169&0.2567 & 0.8743 & 0.9081  \\ 
			&CausE & 0.2352 & 0.6045 & 0.8641&\underline{0.2434} & 0.6296 & 0.9161&0.2444 & 0.8776 & 0.9096  \\ 
			&BC\_LOSS & \underline{0.2295} & \underline{0.5977} & \underline{0.8626}&0.2455 & 0.6512 & 0.923&0.2639 & 0.8811 & 0.9147  \\ 
			&Proposed & \textbf{0.2152} & \textbf{0.5678} & \textbf{0.8548} &\textbf{0.2281} & \textbf{0.580}6 & \textbf{0.9018}&\textbf{0.2298} & \textbf{0.8603} & \textbf{0.8959} \\ \hline 
			\multirow{10}[2]{*}{LightGCN} & RNS & 0.2316 & 0.624 & 0.8715 & 0.2495 & 0.6325 & 0.9226&0.2465 & 0.8583 & 0.8948 \\ 
			&FairStatic & 0.2313 & 0.6723 & 0.8743&0.255 & 0.7348 & 0.9344 &0.2461 & 0.856 & 0.893  \\ 
			&DNS & 0.2298 & 0.6106 & 0.8666&0.2478 & 0.6125 & 0.9118   &\underline{0.2437} & \underline{0.8529} & \underline{0.8905}  \\ 
			&PNS & 0.2303 & 0.5955 & 0.8619&0.2424 & 0.6692 & 0.9243 &0.2462 & 0.8704 & 0.9044  \\ 
			&IPS-CN & 0.2349 & 0.663 & 0.8844 &0.2507 & 0.706 & 0.9369 & 0.2577 & 0.8889 & 0.9181\\ 
			&MACR & 0.2372 & 0.6115 & 0.8686 &0.2441 & 0.6129 & 0.9139 &0.2779 & 0.9117 & 0.9351 \\ 
			&SAMREG & 0.2296 & 0.5926 & 0.8599 &\underline{0.2394} & 0.6197 & 0.9129&0.2487 & 0.8609 & 0.8978 \\ 
			&CausE & 0.2446 & \underline{0.5835} & \underline{0.8566}  &0.2417 & \underline{0.6065} & \underline{0.908} &0.2456 & 0.8787 & 0.9092\\ 
			&BC\_LOSS & \underline{0.2182} & 0.6356 & 0.8701 &0.2411 & 0.7918 & 0.9598&0.2483 & 0.8668 & 0.9012 \\ 
			&Proposed & \textbf{0.2175} & \textbf{0.562} & \textbf{0.8516} &\textbf{0.2375}& \textbf{0.5812} &\textbf{0.9008} &\textbf{0.2415} & \textbf{0.8462} & \textbf{0.8848} \\ \bottomrule
		\end{tabular}%
	}
	\label{tab:MSEbias}%
\end{table}%
\subsection{Experimental results}
Table~\ref{tab:msetpr} presents the performance of MSE, FPR, and FNR. MSE measures the model's performance, where the prediction is calculated by applying a sigmoid function after normalizing the scores. Meanwhile, FPR evaluates the overestimation rate of the top-5 ranking list, while FNR evaluates the underestimation rate of the top-5 ranking list. Both FPR, FNR metrics reflect the overall debiasing performance. The overall ranking performance and popularity debiasing performance of top-10 and top-20 are presented in Table~\ref{tab:performance10} and Table~\ref{tab:performance20}, respectively. 

\begin{table*}[!]
	\centering
	\caption{Overall performance of top-10.}\label{tab:performance10}
	\renewcommand{\arraystretch}{0.8}
	\resizebox{1\textwidth}{!}{
		\begin{tabular}{ccccccccccc}
			\toprule
			\multirow{3}{*}{Dataset} & \multirow{3}{*}{CF Model} & \multirow{3}{*}{Method} & \multicolumn{4}{c}{Ranking Perfomence} & \multicolumn{4}{c}{Popularity Debiasing Performence} \\
			\cmidrule(lr){4-7} \cmidrule(lr){8-11}          &       &       & Pre & Recall & F1 & NDCG & OHR & OCR & UCR & UHR \\
			\midrule
			\multirow{20}[4]{*}{MovieLens-100K} & \multirow{10}[2]{*}{MF} & RNS   & 0.3074 & 0.1992 & 0.2003 & 0.3698 & 0.6924 & 0.0101 & 0.9286 & 0.6895 \\
			&       & FairStatic & 0.2804 & 0.2033 & 0.1986 & 0.351 & 0.5846 & 0.6445 & 0.8737 & 0.7212 \\
			&       & DNS   & 0.3241 & 0.2114 & 0.2124 & 0.3865 & 0.6756 & 0.0138 & 0.9284 & 0.6694 \\
			&       & PNS   & \underline{0.3415}  & 0.218 & 0.223 & 0.4026 & 0.6443 & 0.289 & 0.9105 & 0.6761 \\
			&       & IPS-CN & 0.2472 & 0.1462 & 0.1487 & 0.1538 & 0.6695 & 0.5888 & 0.8746 & 0.8277 \\
			&       & MACR  & 0.335 & 0.2206 & 0.2226 & 0.3977 & 0.6522 & 0.259 & 0.9071 & 0.6744 \\
			&       & SAMREG & 0.3363 & 0.2201 & 0.2228 & 0.3986 & 0.6447 & 0.3201 & 0.8946 & 0.6855 \\
			&       & CausE & 0.3356 & 0.2196 & 0.2223 & 0.3968 & 0.6478 & 0.2988 & 0.9064 & 0.676 \\
			&       & BC\_LOSS & 0.3388 & \underline{0.2219}  & \underline{0.2243} & \underline{0.40}5 & 0.6507 & 0.2195 & 0.9151 & 0.6635 \\
			&       & \textbf{Proposed} & \textbf{0.3584} & \textbf{0.2295} & \textbf{0.2342} & \textbf{0.4307} & 0.6407 & 0.0374 & 0.9254 & 0.64 \\
			\cmidrule{2-11}          & \multirow{10}[2]{*}{LightGCN} & RNS   & 0.3194 & 0.2096 & 0.2103 & 0.3805 & 0.6801 & 0.0182 & 0.9285 & 0.6715 \\
			&       & FairStatic & 0.3077 & 0.2117 & 0.2109 & 0.3753 & 0.6199 & 0.5816 & 0.8886 & 0.7015 \\
			&       & DNS   & 0.3303 & 0.2179 & 0.2186 & 0.395 & 0.6695 & 0.026 & 0.9269 & 0.6595 \\
			&       & PNS   & 0.3487 & 0.227 & 0.2301 & 0.4119 & 0.6411 & 0.2335 & 0.9103 & 0.6652 \\
			&       & IPS-CN & 0.2836 & 0.1846 & 0.1863 & 0.3386 & 0.6758 & 0.493 & 0.891 & 0.7422 \\
			&       & MACR  & 0.339 & 0.2191 & 0.2221 & 0.4002 & 0.6532 & 0.1713 & 0.9178 & 0.6668 \\
			&       & SAMREG & 0.3439 & 0.2236 & 0.2267 & 0.4082 & 0.637 & 0.3318 & 0.896 & 0.6765 \\
			&       & CausE & \underline{0.3536} & \underline{0.2318} & \underline{0.2344} & \underline{0.42}  & 0.6264 & 0.3285 & 0.9   & 0.6611 \\
			&       & BC\_LOSS & 0.308 & 0.2073 & 0.2062 & 0.3709 & 0.6917 & 0.0321 & 0.9269 & 0.6738 \\
			&       & \textbf{Proposed} & \textbf{0.3657} & \textbf{0.2339} & \textbf{0.238}9 & \textbf{0.4377} & 0.6332 & 0.0711 & 0.9233 & 0.6328 \\
			\midrule
			\multirow{20}[4]{*}{yahoo} & \multirow{10}[2]{*}{MF} & RNS   & 0.1029 & 0.1517 & 0.117 & 0.1575 & 0.8915 & 0.1893 & 0.9776 & 0.6179 \\
			&       & FairStatic & 0.0997 & 0.1491 & 0.1142 & 0.1567 & 0.8887 & 0.4509 & 0.9747 & 0.6291 \\
			&       & DNS   & \textbf{0.1053} & \textbf{0.1559} &\textbf{0.1199}  & \underline{0.1624} & 0.8913 & 0.1632 & 0.9788 & 0.6069 \\
			&       & PNS   & 0.0933 & 0.1374 & 0.1062 & 0.1453 & 0.8944 & 0.3864 & 0.9749 & 0.6527 \\
			&       & IPS-CN & 0.0289 & 0.0372 & 0.0305 & 0.0368 & 0.5813 & 0.9792 & 0.9504 & 0.9038 \\
			&       & MACR  & 0.1026 & 0.1516 & 0.1168 & 0.1581 & 0.8954 & 0.1664 & 0.9799 & 0.6154 \\
			&       & SAMREG & 0.0982 & 0.1441 & 0.1114 & 0.1458 & 0.8947 & 0.2626 & 0.9753 & 0.6362 \\
			&       & CausE & 0.0967 & 0.142 & 0.1098 & 0.1447 & 0.8999 & 0.142 & 0.9816 & 0.6324 \\
			&       & BC\_LOSS & 0.1012 & 0.1459 & 0.1139 & 0.1408 & 0.8881 & 0.3374 & 0.971 & 0.6424 \\
			&       & \textbf{Proposed} & \underline{0.1047} & \underline{0.1554} & \underline{0.1194} & \textbf{0.1638} & 0.8948 & 0.0866 & 0.9828 & 0.5994 \\
			\cmidrule{2-11}          & \multirow{10}[2]{*}{LightGCN} & RNS   & 0.1081 & 0.1594 & 0.123 & 0.1664 & 0.8813 & 0.2029 & 0.9741 & 0.6056 \\
			&       & FairStatic & 0.108 & 0.1597 & 0.1231 & 0.1673 & 0.8718 & 0.4748 & 0.9653 & 0.6178 \\
			&       & DNS   & \underline{0.1109} & \underline{0.1643} & \underline{0.1265} & \underline{0.172} & 0.8801 & 0.2164 & 0.9722 & 0.5976 \\
			&       & PNS   & 0.0975 & 0.1433 & 0.1109 & 0.1514 & 0.8853 & 0.4443 & 0.9682 & 0.649 \\
			&       & IPS-CN & 0.0839 & 0.1236 & 0.0954 & 0.129 & 0.9087 & 0.1088 & 0.9802 & 0.6785 \\
			&       & MACR  & 0.0682 & 0.0993 & 0.0774 & 0.1045 & 0.7685 & 0.5147 & 0.9835 & 0.7331 \\
			&       & SAMREG & 0.1069 & 0.1566 & 0.1213 & 0.1614 & 0.8804 & 0.2886 & 0.9703 & 0.6175 \\
			&       & CausE & 0.0951 & 0.143 & 0.1092 & 0.1465 & 0.9042 & 0.0506 & 0.9835 & 0.9835 \\
			&       & BC\_LOSS & 0.1029 & 0.1527 & 0.1174 & 0.1572 & 0.8926 & 0.1092 & 0.9807 & 0.6131 \\
			&       & \textbf{Proposed} & \textbf{0.1159} & \textbf{0.1713} & \textbf{0.132} & \textbf{0.1822} & 0.8814 & 0.2486 & 0.9724 & 0.5799 \\
			\midrule
			\multicolumn{1}{c}{\multirow{20}[4]{*}{1M}} & \multirow{10}[2]{*}{MF} & RNS   & 0.3124 & 0.124 & 0.1479 & 0.3523 & 0.6871 & 0.0234 & 0.9484 & 0.8076 \\
			&       & FairStatic & 0.2919 & 0.1425 & 0.1606 & 0.34  & 0.5923 & 0.6155 & 0.8989 & 0.8174 \\
			&       & DNS   & 0.3309 & 0.1359 & 0.1605 & 0.3723 & 0.6687 & 0.0348 & 0.9474 & 0.7897 \\
			&       & PNS   & \underline{0.3362} & \underline{0.1542} & \underline{0.1759} & \underline{0.3802} & 0.6472 & 0.3165 & 0.9253 & 0.7835 \\
			&       & IPS-CN & 0.2128 & 0.0841 & 0.1007 & 0.2356 & 0.6905 & 0.4761 & 0.9219 & 0.8948 \\
			&       & MACR  & 0.3332 & 0.1369 & 0.162 & 0.3738 & 0.6664 & 0.0477 & 0.9466 & 0.7882 \\
			&       & SAMREG & 0.3286 & 0.1382 & 0.1621 & 0.3689 & 0.6698 & 0.0569 & 0.9452 & 0.7888 \\
			&       & CausE & 0.3218 & 0.131 & 0.155 & 0.3629 & 0.6775 & 0.0283 & 0.9478 & 0.797 \\
			&       & BC\_LOSS & 0.3033 & 0.1262 & 0.1488 & 0.3367 & 0.6955 & 0.0758 & 0.946 & 0.8039 \\
			&       & \textbf{Proposed} & \textbf{0.3618} & \textbf{0.1589} & \textbf{0.184} & \textbf{0.4074} & 0.6383 & 0.064 & 0.9436 & 0.7561 \\
			\cmidrule{2-11}          & \multirow{10}[2]{*}{LightGCN} & RNS   & 0.3296 & 0.1418 & 0.1649 & 0.368 & 0.6595 & 0.2643 & 0.928 & 0.7977 \\
			&       & FairStatic & 0.2299 & 0.1088 & 0.1225 & 0.2603 & 0.4813 & 0.8095 & 0.8585 & 0.8973 \\
			&       & DNS   & 0.3359 & 0.1446 & 0.1683 & 0.3757 & 0.655 & 0.2629 & 0.9269 & 0.7931 \\
			&       & PNS   & 0.2866 & 0.1252 & 0.1448 & 0.3211 & 0.6563 & 0.5157 & 0.9089 & 0.84 \\
			&       & IPS-CN & 0.2501 & 0.1031 & 0.1224 & 0.2828 & 0.6859 & 0.4773 & 0.9219 & 0.8648 \\
			&       & MACR  & 0.3323 & 0.1408 & 0.1649 & 0.3735 & 0.6647 & 0.1158 & 0.9446 & 0.7848 \\
			&       & SAMREG & 0.3326 & 0.1464 & 0.1464 & 0.3701 & 0.6454 & 0.4073 & 0.9133 & 0.8015 \\
			&       & CausE & \underline{0.34}  & \underline{0.1511} & \underline{0.1737} & \underline{0.381} & 0.6452 & 0.3329 & 0.9191 & 0.7901 \\
			&       & BC\_LOSS & 0.1822 & 0.0689 & 0.082 & 0.2025 & 0.8172 & 0.0007 & 0.951 & 0.8952 \\
			&       & \textbf{Proposed} & \textbf{0.3612} & \textbf{0.1622} & \textbf{0.186} & \textbf{0.4063} & 0.633 & 0.2676 & 0.9246 & 0.7654 \\
			\bottomrule
		\end{tabular}
	}
\end{table*}%

\begin{table*}[!]
	\centering
	\caption{Overall performance of top-20.}\label{tab:performance20}
	\renewcommand{\arraystretch}{0.8}
	\resizebox{1\textwidth}{!}{
		\begin{tabular}{ccccccccccc}
			\toprule
			\multirow{3}{*}{Dataset} & \multirow{3}{*}{CF Model} & \multirow{3}{*}{Method} & \multicolumn{4}{c}{Ranking Perfomence} & \multicolumn{4}{c}{Popularity Debiasing Performence} \\
			\cmidrule(lr){4-7} \cmidrule(lr){8-11}          &       &       & Pre & Recall & F1 & NDCG & OHR & OCR & UCR & UHR \\
			\midrule
			\multirow{20}[4]{*}{MovieLens-100K} & \multirow{10}[2]{*}{MF} & RNS   & 0.2478 & 0.3006 & 0.2236 & 0.3673 & 0.7505 & 0.0552 & 0.9258 & 0.5285 \\
			&       & FairStatic & 0.2156 & 0.2955 & 0.2078 & 0.3462 & 0.6323 & 0.8746 & 0.7932 & 0.6273 \\
			&       & DNS   & 0.259 & 0.3141 & 0.234 & 0.3822 & 0.7387 & 0.1087 & 0.9218 & 0.5108 \\
			&       & PNS   & \underline{0.2741} & \underline{0.3373} & \underline{0.2511} & 0.4016 & 0.7031 & 0.5653 & 0.8649 & 0.5215 \\
			&       & IPS-CN & 0.2067 & 0.2379 & 0.1833 & 0.284 & 0.7191 & 0.7356 & 0.8145 & 0.713 \\
			&       & MACR  & 0.2695 & 0.336 & 0.2478 & 0.3968 & 0.7162 & 0.4723 & 0.8781 & 0.5157 \\
			&       & SAMREG & 0.271 & 0.3372 & 0.2491 & 0.3978 & 0.7082 & 0.5221 & 0.8585 & 0.5293 \\
			&       & CausE & 0.2701 & 0.3371 & 0.2485 & 0.3977 & 0.7135 & 0.4631 & 0.8685 & 0.5195 \\
			&       & BC\_LOSS & 0.2707 & 0.3359 & 0.2484 & \underline{0.4029} & 0.7173 & 0.4375 & 0.881 & 0.5101 \\
			&       & \textbf{Proposed} & \textbf{0.2855} & \textbf{0.3443} & \textbf{0.2596} & \textbf{0.4236} & 0.7098 & 0.2126 & 0.914 & 0.4668 \\
			\cmidrule{2-11}          & \multirow{10}[2]{*}{LightGCN} & RNS   & 0.2537 & 0.3116 & 0.2303 & 0.3761 & 0.7434 & 0.1503 & 0.9218 & 0.5154 \\
			&       & FairStatic & 0.2099 & 0.2912 & 0.2035 & 0.3336 & 0.625 & 0.883 & 0.7784 & 0.6463 \\
			&       & DNS   & 0.2619 & 0.3213 & 0.2383 & 0.387 & 0.7345 & 0.1773 & 0.9175 & 0.5005 \\
			&       & PNS   & 0.2758 & 0.3408 & 0.2528 & 0.4068 & 0.7048 & 0.5365 & 0.8693 & 0.5119 \\
			&       & IPS-CN & 0.2261 & 0.2854 & 0.2083 & 0.3364 & 0.7263 & 0.6784 & 0.8375 & 0.6177 \\
			&       & MACR  & 0.2729 & 0.3355 & 0.249 & 0.4002 & 0.7145 & 0.4271 & 0.8861 & 0.5044 \\
			&       & SAMREG & 0.2784 & 0.3462 & 0.2552 & 0.4099 & 0.6985 & 0.5397 & 0.8443 & 0.5208 \\
			&       & CausE & \underline{0.2837} & \textbf{0.3548} & \underline{0.2612} & \underline{0.4199} & 0.6918 & 0.5827 & 0.8429 & 0.5066 \\
			&       & BC\_LOSS & 0.2509 & 0.3158 & 0.2308 & 0.3746 & 0.7474 & 0.1322 & 0.9186 & 0.5066 \\
			&       & \textbf{Proposed} & \textbf{0.2908} & \underline{0.3524} & \textbf{0.2642} & \textbf{0.4317} & 0.7051 & 0.2778 & 0.9075 & 0.4585 \\
			\midrule
			\multirow{20}[4]{*}{yahoo} & \multirow{10}[2]{*}{MF} & RNS   & 0.0755 & 0.2215 & 0.1078 & 0.1854 & 0.9204 & 0.2842 & 0.9688 & 0.4815 \\
			&       & FairStatic & 0.069 & 0.2059 & 0.0993 & 0.1785 & 0.9114 & 0.9183 & 0.9522 & 0.5358 \\
			&       & DNS   & \underline{0.0769} & \textbf{0.2262} & \textbf{0.1101} & \textbf{0.1906} & 0.9193 & 0.3849 & 0.9656 & 0.4748 \\
			&       & PNS   & 0.0672 & 0.1959 & 0.0958 & 0.168 & 0.9201 & 0.7036 & 0.9592 & 0.5485 \\
			&       & IPS-CN & 0.0253 & 0.066 & 0.0346 & 0.0488 & 0.7813 & 0.9824 & 0.9251 & 0.8715 \\
			&       & MACR  & 0.0752 & 0.2206 & 0.1075 & 0.1857 & 0.9225 & 0.3764 & 0.9702 & 0.4806 \\
			&       & SAMREG & 0.0737 & 0.216 & 0.1053 & 0.1741 & 0.9204 & 0.327 & 0.9674 & 0.4966 \\
			&       & CausE & 0.0708 & 0.2077 & 0.1012 & 0.171 & 0.9256 & 0.3281 & 0.9737 & 0.5038 \\
			&       & BC\_LOSS & \textbf{0.0773} & \underline{0.2239} & \underline{0.1099} & 0.1728 & 0.9157 & 0.4967 & 0.5044 & 0.4967 \\
			&       & \textbf{Proposed} & 0.0753 & 0.2221 & 0.1079 & \underline{0.1905} & 0.9237 & 0.2744 & 0.9729 & 0.4717 \\
			\cmidrule{2-11}          & \multirow{10}[2]{*}{LightGCN} & RNS   & 0.0783 & 0.2291 & 0.1118 & 0.1941 & 0.9129 & 0.3962 & 0.9574 & 0.4797 \\
			&       & FairStatic & 0.0776 & 0.2276 & 0.1109 & 0.1942 & 0.9019 & 0.8021 & 0.9361 & 0.5119 \\
			&       & DNS   & \underline{0.0806} & \underline{0.2358} & \underline{0.115} & \underline{0.2009} & 0.9118 & 0.4293 & 0.9546 & 0.469 \\
			&       & PNS   & 0.0695 & 0.2019 & 0.099 & 0.1742 & 0.8966 & 0.9304 & 0.9211 & 0.6073 \\
			&       & IPS-CN & 0.0632 & 0.1874 & 0.0907 & 0.1532 & 0.9304 & 0.1297 & 0.9368 & 0.5532 \\
			&       & MACR  & 0.0399 & 0.1163 & 0.0573 & 0.1095 & 0.7748 & 0.9909 & 0.9777 & 0.7057 \\
			&       & SAMREG & 0.0771 & 0.2244 & 0.1099 & 0.1882 & 0.911 & 0.4855 & 0.9522 & 0.4958 \\
			&       & CausE & 0.0698 & 0.2088 & 0.1005 & 0.1732 & 0.9293 & 0.0786 & 0.9803 & 0.4927 \\
			&       & BC\_LOSS & 0.0749 & 0.2211 & 0.1072 & 0.1847 & 0.9201 & 0.1794 & 0.9731 & 0.4761 \\
			&       & \textbf{Proposed} & \textbf{0.0824} & \textbf{0.2422} & \textbf{0.118} & \textbf{0.2099} & 0.9142 & 0.5388 & 0.9555 & 0.4519 \\
			\midrule
			\multicolumn{1}{c}{\multirow{20}[4]{*}{1M}} & \multirow{10}[2]{*}{MF} & RNS   & 0.257 & 0.1928 & 0.178 & 0.3287 & 0.7424 & 0.0638 & 0.9449 & 0.7012 \\
			&       & FairStatic & 0.2309 & 0.214 & 0.1814 & 0.3181 & 0.6433 & 0.8103 & 0.8419 & 0.744 \\
			&       & DNS   & 0.2726 & 0.2105 & 0.192 & 0.3497 & 0.7263 & 0.0921 & 0.9413 & 0.6773 \\
			&       & PNS   & \underline{0.2773} & \underline{0.2386} & \underline{0.2079} & \underline{0.3639} & 0.6999 & 0.5152 & 0.8916 & 0.6765 \\
			&       & IPS-CN & 0.1791 & 0.1343 & 0.1242 & 0.2236 & 0.7524 & 0.551 & 0.9053 & 0.8308 \\
			&       & MACR  & 0.2746 & 0.2124 & 0.1939 & 0.3514 & 0.7245 & 0.1224 & 0.9399 & 0.6749 \\
			&       & SAMREG & 0.2709 & 0.2145 & 0.1934 & 0.3489 & 0.7271 & 0.1272 & 0.9383 & 0.6752 \\
			&       & CausE & 0.2649 & 0.2027 & 0.1854 & 0.3403 & 0.7345 & 0.0722 & 0.9433 & 0.6874 \\
			&       & BC\_LOSS & 0.2513 & 0.2007 & 0.1807 & 0.3195 & 0.7474 & 0.1552 & 0.9378 & 0.6926 \\
			&       & \textbf{Proposed} & \textbf{0.2985} & \textbf{0.2456} & \textbf{0.2185} & \textbf{0.3878} & 0.7006 & 0.1842 & 0.9305 & 0.6283 \\
			\cmidrule{2-11}          & \multirow{10}[2]{*}{LightGCN} & RNS   & 0.2738 & 0.2225 & 0.1984 & 0.351 & 0.7132 & 0.4206 & 0.9046 & 0.6889 \\
			&       & FairStatic & 0.1942 & 0.1756 & 0.1491 & 0.2546 & 0.5605 & 0.8676 & 0.7897 & 0.8379 \\
			&       & DNS   & 0.2806 & 0.2298 & 0.2042 & 0.3604 & 0.707 & 0.4405 & 0.9015 & 0.6788 \\
			&       & PNS   & 0.2385 & 0.2009 & 0.1754 & 0.3074 & 0.7018 & 0.6573 & 0.8708 & 0.7558 \\
			&       & IPS-CN & 0.2045 & 0.1632 & 0.1473 & 0.2654 & 0.7366 & 0.6386 & 0.8932 & 0.7925 \\
			&       & MACR  & 0.2757 & 0.2203 & 0.1985 & 0.3536 & 0.7201 & 0.3132 & 0.933 & 0.6685 \\
			&       & SAMREG & 0.2798 & 0.2337 & 0.2057 & 0.3574 & 0.6961 & 0.5911 & 0.8771 & 0.6894 \\
			&       & CausE & \underline{0.2836} & \underline{0.2386} & \underline{0.2093} & \underline{0.3664} & 0.6989 & 0.53  & 0.8841 & 0.6777 \\
			&       & BC\_LOSS & 0.1507 & 0.1128 & 0.1025 & 0.1902 & 0.8486 & 0.0007 & 0.951 & 0.8285 \\
			&       & \textbf{Proposed} & \textbf{0.3} & \textbf{0.2528} & \textbf{0.2217} & \textbf{0.39} & 0.6928 & 0.4339 & 0.8977 & 0.6437 \\
			\bottomrule
			
		\end{tabular}
	}
\end{table*}%

\newpage

\section{Proofs}
\subsection{Proof of Proposition~\ref{theorem:bias_mse}}\label{appendix:bias_mse}
\begin{proposition}
	Denote $y^* =\mathbb E_y [y|\mathbf{x}]$ as the Bayesian optimal prediction over all measurable functions. For  fixed complexity of the functional class $\mathcal{G}$ and joint distribution $p(x,y)$ on $\mathcal{Z}=\mathcal{X}\times\mathcal{Y}$, assuming the noise $y^*-y$ is zero-mean, then
	\begin{eqnarray}
		\textsc{Bias}_\mathcal{D} &=& \textsc{Mse}_{\mathcal{D}}  - const  \nonumber
	\end{eqnarray}
\end{proposition}
\begin{proof}
	We first compute the squared error of single data point $\mathbf{x}$ over different training set $\mathcal{D}$
	\begin{eqnarray}
		\label{eq:first_seg}
		&&	\int_{\mathcal{D}}(g(\mathbf{x};\mathcal{D}) - y)^2p(\mathcal{D})d\mathcal{D} \label{eq:se}\\
		&=& \int_{\mathcal{D}}\{g(\mathbf{x};\mathcal{D}) - \mathbb{E}_\mathcal{D}[g(\mathbf{x};\mathcal{D})] + \mathbb{E}_\mathcal{D}[g(\mathbf{x};\mathcal{D})] - y\}^2p(\mathcal{D})d\mathcal{D} \nonumber \\
		&=&  \int_{\mathcal{D}}\{\{g(\mathbf{x};\mathcal{D}) - \mathbb{E}_\mathcal{D}[g(\mathbf{x};\mathcal{D})] \}^2p(\mathcal{D}) + \{\mathbb{E}_\mathcal{D}[g(\mathbf{x};\mathcal{D})] - y\}^2p(\mathcal{D}) \nonumber \\
		&&+ 2\{g(\mathbf{x};\mathcal{D}) - \mathbb{E}_\mathcal{D}[g(\mathbf{x};\mathcal{D})]\}\{\mathbb{E}_\mathcal{D}[g(\mathbf{x};\mathcal{D})] - y\}p(\mathcal{D})\}d\mathcal{D} \nonumber \\
		&=& \int_{\mathcal{D}}\{\{g(\mathbf{x};\mathcal{D}) - \mathbb{E}_\mathcal{D}[g(\mathbf{x};\mathcal{D})] \}^2p(\mathcal{D}) + \{\mathbb{E}_\mathcal{D}[g(\mathbf{x};\mathcal{D})] - y\}^2p(\mathcal{D}) \}d\mathcal{D}. \label{eq:sepoint}
	\end{eqnarray}
	Furthermore, by taking the expectation of the square error in Eq~\eqref{eq:se} over all sample $\mathbf{x}$, we obtain the general mean squared error $\textsc{Mse}_{\mathcal{D}}$
	\begin{eqnarray}
		\label{eq:first_result}
		\textsc{Mse}_{\mathcal{D}}&=&\int_{\mathbf{x}\in \mathcal{X}}p(\mathbf{x},y)\int_{\mathcal{D}}(g(\mathbf{x};\mathcal{D}) - y)^2p(\mathcal{D})d\mathbf{x}d\mathcal{D}	\nonumber \\
		&=&\int_{\mathbf{x}\in \mathcal{X}}p(\mathbf{x},y)\int_{\mathcal{D}}\{g(\mathbf{x};\mathcal{D}) - \mathbb{E}_\mathcal{D}[g(\mathbf{x};\mathcal{D})] \}^2p(\mathcal{D})d\mathcal{D}d\mathbf{x} \label{eq:varians}\\
		&&+\int_{\mathbf{x}\in \mathcal{X}}p(\mathbf{x},y)\int_{\mathcal{D}}\{\mathbb{E}_\mathcal{D}[g(\mathbf{x};\mathcal{D})] - y\}^2p(\mathcal{D})d\mathcal{D}d\mathbf{x} \label{eq:next_seg}
	\end{eqnarray}
	It represents the expected squared error over all training datasets and all data points, and measures the model performance.  $\textsc{Mse}_{\mathcal{D}} \rightarrow 0$ as $g(\mathbf{x};\mathcal{D}) = y$ almost surely for any $\mathbf{x}$ and $\mathcal{D}$. Next we introduce $y^*$ to further decompose Eq~\eqref{eq:next_seg}, where $y^*=\mathbb E_y [y|\mathbf{x}]$ is the expectation of ground true label $y$ under the posterior distribution $p(y|\mathbf{x})$, corresponding to the Bayesian optimal prediction over all measurable functions.
	
	\begin{eqnarray}
		\label{eq:second_seg}
		&&\int_{\mathbf{x}\in \mathcal{X}}p(\mathbf{x},y)\int_{\mathcal{D}}\{\mathbb{E}_\mathcal{D}[g(\mathbf{x};\mathcal{D})] - y\}^2p(\mathcal{D})d\mathcal{D}d\mathbf{x} \nonumber \\
		&=& \int_{\mathbf{x}\in \mathcal{X}}p(\mathbf{x},y)\int_{\mathcal{D}}\{\mathbb{E}_\mathcal{D}[g(\mathbf{x};\mathcal{D})] - y^* + y^* - y\}^2p(\mathcal{D})d\mathcal{D}d\mathbf{x} \nonumber \\
		&=& \int_{\mathbf{x}\in \mathcal{X}}p(\mathbf{x},y)\int_{\mathcal{D}}\{\{\mathbb{E}_\mathcal{D}[g(\mathbf{x};\mathcal{D})] - y^*\}^2 + \{y^*-y\}^2 \nonumber \\
		&&+ 2\{\mathbb{E}_\mathcal{D}[g(\mathbf{x};\mathcal{D})] - y^*\}\{y^*-y\}p(\mathcal{D})\}d\mathcal{D}d\mathbf{x} \nonumber \\
		&=& \int_{\mathbf{x}\in \mathcal{X}}p(\mathbf{x},y)\int_{\mathcal{D}}\{\{\mathbb{E}_\mathcal{D}[g(\mathbf{x};\mathcal{D})] - y^*\}^2 + \{y^*-y\}^2 \}d\mathcal{D}d\mathbf{x} \label{eq:varnoisy}
	\end{eqnarray}
	Equation~\eqref{eq:varnoisy} is obtained by assuming that the noise $y^*-y$ is zero-mean, that is, $\int_{x\in \mathcal{X}}p(\mathbf{x},y)(y^*-y)d\mathbf{x} = 0$. By substituting Equation~\eqref{eq:varnoisy} back into Equation~\eqref{eq:next_seg}, we obtain the following expression
	\begin{eqnarray}
		\label{eq:second_result}
		\textsc{Mse}_{\mathcal{D}}	&=& \int_{\mathbf{x}\in \mathcal{X}}p(\mathbf{x},y)\int_{\mathcal{D}}\{g(\mathbf{x};\mathcal{D}) - \mathbb{E}_\mathcal{D}[g(\mathbf{x};\mathcal{D})] \}^2p(\mathcal{D})d\mathcal{D}d\mathbf{x} \label{eq:variance1}\\
		&+&\int_{\mathbf{x}\in \mathcal{X}}p(\mathbf{x},y)\int_{\mathcal{D}}\{\mathbb{E}_\mathcal{D}[g(\mathbf{x};\mathcal{D})] - y^*\}^2d\mathcal{D}d\mathbf{x} \label{eq:biase1}\\
		&+&\int_{\mathbf{x}\in \mathcal{X}}p(\mathbf{x},y)\int_{\mathcal{D}}\{y^*-y\}^2d\mathcal{D}d\mathbf{x}\label{eq:noise}
	\end{eqnarray}
	Eq~\eqref{eq:variance1} is variance measures the extent to which the prediction on individual data sets vary around their average. For fixed complexity of the functional class $\mathcal{G}$ and joint distribution $p(x,y)$ on $\mathcal{Z}=\mathcal{X}\times\mathcal{Y}$, the variance term is fixed. 
	
	Eq~\eqref{eq:biase1} is squared bias $\textsc{Bias}_\mathcal{D}$ that measures the extent to which the average prediction over all data sets differs from the optimal prediction. This is another formal expression of the popularity bias, which refers to the tendency of collaborative filtering models to recommend popular items excessively, as observed in many datasets. 
	
	Eq~\eqref{eq:noise} measures noisy, for fixed joint distribution $p(x,y)$ on $\mathcal{Z}=\mathcal{X}\times\mathcal{Y}$, it also fixed. So we have:
	\[\textsc{Mse}_{\mathcal{D}} = \textsc{Bias}_\mathcal{D}+const\]
	which completes the proof.
\end{proof}

\subsection{Proof of Proposition~\ref{theorem:metric}}\label{appendix:metric}
\begin{lemma}
	Let $g(\mathbf{x}) \in \{0,1\}$ and $y \in \{0,1\}$, that is, simplifying implicit collaborative filtering to a classification problem. With the result of Proposition~\ref{theorem:bias_mse}, we have
	\begin{eqnarray}
		bias \simeq \textsc{FPR}+\textsc{FNR}
	\end{eqnarray}
	\begin{proof}
		We apply the countable additivity property of Lebesgue integral to complete the proof. By setting $g(\mathbf{x}) \in \{0,1\}$ and $y \in \{0,1\}$, we can rewrite the integrand term as:
		\begin{eqnarray}\label{eq:interrand2}
			(g(\mathbf{x}) - y)^2 = \mathbb{I}(g(\mathbf{x})>y) + \mathbb{I}(g(\mathbf{x})<y)
		\end{eqnarray}
		where $\mathbb{I}(\cdot)$ is the indicator function, $\mathbb{I}(\cdot)=1$ if the augment is true otherwise 0.
		Then we insert Eq~\ref{eq:interrand2} back to  Eq~\ref{eq:biaseq} and omit the constant in Eq~\ref{eq:biaseq}, we have
		\begin{eqnarray}
			bias &\simeq&  \int_{\mathbf{x}\in {\mathcal{X}}} (g(\mathbf{x}) - y)^2 p(\mathbf{x}) d\mathbf{x}  \nonumber\\
			&=& \int_{\mathbf{x}\in {\mathcal{X}}}  [\mathbb{I}(g(\mathbf{x})>y) + \mathbb{I}(g(\mathbf{x})<y)] p(\mathbf{x}) d\mathbf{x}\nonumber\\
			&=& \int_{\mathbf{x}\in {\mathcal{X}}}  \mathbb{I}(g(\mathbf{x})>y)p(\mathbf{x}) d\mathbf{x} +\int_{\mathbf{x}\in {\mathcal{X}}} \mathbb{I}(g(\mathbf{x})<y) p(\mathbf{x}) d\mathbf{x} \label{eq:2subset}
		\end{eqnarray}
		The first term of Eq~\ref{eq:2subset} is the expectation of $\mathbb{I}(g(\mathbf{x})>y)$, which admits the empirical FPR estimate over the set of negative samples $\mathcal{S}_n\subset \mathcal{X}$: 
		\begin{eqnarray}
			\int_{\mathbf{x}\in \mathcal{X}}\mathbb{I}(g(\mathbf{x})>y) p(\mathbf{x}) d\mathbf{x} &=& \mathbb{E}_{\mathbf{x}} \mathbb{I}(g(\mathbf{x})>y) \nonumber \\ 
			&\simeq&	\frac{ \sum_{\mathbf{x}\in  \mathcal{S}_n} \mathbb I(g(\mathbf{x})>y)}{|\mathcal{S}_n|} \nonumber \\
			&=& \frac{\sum_{\mathbf{x}\in  \mathcal{S}_n}\mathbb I(g(\mathbf{x}) =1, y=0 )}{|\mathcal{S}_n|} \nonumber \\
			&=& FPR
		\end{eqnarray}
		Likewise, the second term of Eq~\ref{eq:2subset}
		\begin{eqnarray}
			\int_{\mathbf{x}\in \mathcal{X}}\mathbb{I}(g(\mathbf{x})<y) p(\mathbf{x}) d\mathbf{x} &=& \mathbb{E}_{\mathbf{x}} \mathbb{I}(g(\mathbf{x})<y) \nonumber \\ 
			&\simeq&	\frac{ \sum_{\mathbf{x}\in  \mathcal{S}_p} \mathbb I(g(\mathbf{x})<y)}{|\mathcal{S}_p|} \nonumber \\
			&=& \frac{\sum_{\mathbf{x}\in  \mathcal{S}_p}\mathbb I(g(\mathbf{x}) =0, y=1 )}{|\mathcal{S}_p|} \nonumber \\
			&=& FNR
		\end{eqnarray}
		which completes the proof. Furthermore, by using the similar algebraic operations as in Eq~\ref{eq:interrand2}, we can rewrite the integrand as follows:
		\begin{eqnarray}\label{eq:interrand4}
			(g(\mathbf{x}) - y)^2 = &&\mathbb{I}(g(\mathbf{x})>y, \mathbf{x}\in \mathcal{X}_\textsc{Hot}) + \mathbb{I}(g(\mathbf{x})>y, \mathbf{x}\in \mathcal{X}_\textsc{Cold}) \nonumber\\
			&+& \mathbb{I}(g(\mathbf{x})<y, \mathbf{x}\in \mathcal{X}_\textsc{Hot}) + \mathbb{I}(g(\mathbf{x})<y, \mathbf{x}\in \mathcal{X}_\textsc{Cold}) 
		\end{eqnarray}
		where $\mathcal{X}_\textsc{Hot} = \{\mathbf{x}|(u,i) \in \mathcal{X}, pop_i \geq \tau  \} $ is the set of samples with item popularity greater than a predefined threshold $\tau$ according to specific scenarios,  $\mathcal{X}_\textsc{Cold} = \mathcal{X} - \mathcal{X}_\textsc{Hot}$.
		So the integrating the first term of Eq~\eqref{eq:interrand4}
		\begin{eqnarray}
			\int_{\mathbf{x}\in \mathcal{X}}\mathbb{I}(g(\mathbf{x})>y, \mathbf{x}\in \mathcal{X}_\textsc{Hot}) p(\mathbf{x}) d\mathbf{x} &=& \mathbb{E}_{\mathbf{x}} \mathbb{I}(g(\mathbf{x})>y, \mathbf{x}\in \mathcal{X}_\textsc{Hot}) \label{eq:hotoverest}
		\end{eqnarray}
		Since only top-K recommended items $\mathcal{S}_\textsc{Rec}$ and popular items $\mathcal{S}_\textsc{Hot}$ are concerned, we take $\mathcal{S}_1 = \mathcal{S}_\textsc{Rec} \cap \mathcal{S}_\textsc{Hot}$ to obtain the empirical estimate of Eq~\ref{eq:hotoverest} 
		\begin{eqnarray}
			\textsc{OHR} &=&\textsc{H}\mathbb{E}_{\mathbf{x}} \mathbb{I}(g(\mathbf{x})>y, \mathbf{x}\in \mathcal{X}_\textsc{Hot}) \\
			&\simeq&	\frac{ \sum_{\mathbf{x}\in  \mathcal{S}_1} \mathbb I(g(\mathbf{x})>y,\mathbf{x}\in \mathcal{X}_\textsc{Hot})}{|\mathcal{S}_1|} \nonumber \\
			&=& \frac{\sum_{\mathbf{x}\in  \mathcal{S}_1}\mathbb I(g(\mathbf{x}) =1, y=0 , \mathbf{x}\in \mathcal{X}_\textsc{Hot})}{|\mathcal{S}_1|} \nonumber \\
			&=& \frac{\{\mathcal{S}_\textsc{Rec}-\mathcal{S}_\textsc{Test}\} \cap \mathcal{X}_\textsc{Hot} }{|\mathcal{S}_\textsc{Rec} \cap \mathcal{X}_\textsc{Hot}|} \label{eq:ORH}
		\end{eqnarray}
		Here, $\mathcal{S}_\textsc{Rec}-\mathcal{S}_\textsc{Test}$ is the set of items recommended to the user that they do not like, i.e., over-recommended items. So, the OHR metric evaluates \textbf{O}ver-recommended \textbf{H}ot item \textbf{R}ate. Likewise,
		\begin{eqnarray}
			\textsc{UHR} &=& \frac{\{\mathcal{S}_\textsc{Test}-\mathcal{S}_\textsc{Rec}\} \cap \mathcal{X}_\textsc{Hot} }{|\mathcal{S}_\textsc{Rec} \cap \mathcal{X}_\textsc{Hot}|} \label{eq:URH}
		\end{eqnarray}
		Here, $\mathcal{S}_\textsc{Test}-\mathcal{S}_\textsc{Rec}$ is the set of items that the user likes but were not recommended to them, i.e., under-recommended items. So UHR metric evaluates \textbf{U}nder-recommended \textbf{H}ot item \textbf{R}ate. The same analysis applies for cold items:
		\begin{eqnarray}
			\textsc{OCR} &=&\frac{\{\mathcal{S}_\textsc{Rec}-\mathcal{S}_\textsc{Test}\} \cap \mathcal{X}_\textsc{Cold} }{|\mathcal{S}_\textsc{Rec} \cap \mathcal{X}_\textsc{Cold}|} \label{eq:ORC}\\
			\textsc{UCR} &=&\frac{\{\mathcal{S}_\textsc{Test}-\mathcal{S}_\textsc{Rec}\} \cap \mathcal{X}_\textsc{Cold} }{|\mathcal{S}_\textsc{Rec} \cap \mathcal{X}_\textsc{Cold}|} \label{eq:URC}
		\end{eqnarray}
		OCR metric evaluates \textbf{O}ver-recommended \textbf{C}old item \textbf{R}ate,  UCR metric evaluates \textbf{U}nder-recommended \textbf{C}old item \textbf{R}ate.
	\end{proof}
\end{lemma}

\subsection{Proof of Theorem~\ref{theorem:auc}}\label{Appendix:auc}
\begin{theorem}[Partial AUC Maximization]
	Denote the posterior probability of $\mathbf{x}\in \textsc{Tn}$ as $ \mathbb{P}(\textsc{Tn}|\mathbf{x})$. Given training set $\mathcal{D}= \mathcal{D}^+ \cup \mathcal{D}^-$, for each unit decrement $dg(\mathbf{x})=-1$, the following negative sampling strategy  
	\begin{eqnarray}
		\mathbf{x} 
		&=& \arg \max _\mathbf{x} \triangle_{\mathbf{x}^+}\cdot \mathbb{P}(\textsc{Tn}|\mathbf{x}) -\triangle_{\mathbf{x}^-}\cdot \mathbb{P}(\textsc{Fn}|\mathbf{x}) \label{eq:negsampling}
	\end{eqnarray}
	is the optimal sampling rule that maximizes the partial AUC metric over training set $\mathcal{D}$, where $\triangle_{\mathbf{x}^+}=\sum_{\mathbf{x}^+ \in \mathcal{D}^+}[1- \sigma(g(\mathbf{x}^+) - g(\mathbf{x}))]$ and $\triangle_{\mathbf{x}^-} = \sum_{\mathbf{x}^-\in \mathcal{N}} [1- \sigma(g(\mathbf{x}) - g(\mathbf{x}^-))]$, and $ \mathcal{N}\subseteq \mathcal{D}^-$ is the set of top-$\lfloor \gamma\cdot|\mathcal{D}^-| \rfloor$ ranked negative instances.
	\begin{proof}
		Following [ref], we replace the indicator function of $0-1$ loss with a differentiable logistic surrogate loss $\ln \sigma(\cdot)$, since the indicator function $\mathbb{I}(\cdot)$ is non-differentiable. Therefor, for any unlabeled instance $\mathbf{x}$, the first order Taylor expansion of pAUC with respect to $g(\mathbf{x})$ is:
		\begin{eqnarray}
			pAUC^{(t+1)} = pAUC^{(t)} + \frac{\partial pAUC}{\partial g(\mathbf{x})} dg(\mathbf{x}) + o (dg(\mathbf{x}) )
		\end{eqnarray}
		The above equation means the increment of  $\triangle pAUC =pAUC^{(t+1)} -pAUC^{(t)} $ caused by a unit change of $dg(\mathbf{x})$ due to the updating of model parameter. Since the label of $\mathbf{x}$ is unknown, there are two cases for the derivative  of pAUC in Eq~\ref{eq:auc}  with respect to $g(\mathbf{x})$, resulting in a piecewise function with respect to its ground truth label
		\begin{equation}
			\frac{\partial pAUC}{\partial g(\mathbf{x})} =\left\{
			\begin{array}{cl}
				-\frac{1}{|\mathcal{D}^+||\mathcal{N}^-|}\sum_{\mathbf{x}^+ \in \mathcal{D}^+} [1- \sigma(g(\mathbf{x}^+) - g(\mathbf{x}))] ,& \mathbf{x} \in \textsc{Tn} \\
				\frac{1}{|\mathcal{D}^+||\mathcal{N}^-|}\sum_{\mathbf{x}^- \in \mathcal{N}^-} [1- \sigma(g(\mathbf{x}) - g(\mathbf{x}^-))] ,&  \mathbf{x} \in \textsc{Fn} \\
			\end{array} \right.
		\end{equation}
		For a unit decrease in $g(\mathbf{x})=-1$, the expected increase of AUC over the ground truth label of $\mathbf{x}$ is:
		\begin{eqnarray}
			\mathbb{E}_{c(\mathbf{x})}\triangle pAUC  &=&  \mathbb{E}_{c(\mathbf{x})} \frac{\partial pAUC}{\partial g(\mathbf{x})} \nonumber \\
			&=& \frac{1}{|\mathcal{D}^+||\mathcal{N}^-|} \sum_{\mathbf{x}^+} [1- \sigma(g(\mathbf{x}^+) - g(\mathbf{x}))] \mathbb{P}(\textsc{Tn}|\mathbf{x})- \sum_{\mathbf{x}^-} [1- \sigma(g(\mathbf{x}) - g(\mathbf{x}^-))] \mathbb{P}(\textsc{Fn}|\mathbf{x})
		\end{eqnarray}
		Given $pAUC^{(t)}$ at $t$ epoch, to maximize the $pAUC^{(t+1)}$ 
		\begin{eqnarray}
			\mathbf{x} 		&=&	\arg \max _\mathbf{x}  pAUC^{(t+1)} \nonumber \\
			&=&	\arg \max _\mathbf{x}  pAUC^{(t+1)} - pAUC^{(t)} \nonumber \\
			&=& \arg \max _\mathbf{x} \mathbb{E}_{c(\mathbf{x})}\triangle pAUC  \nonumber \\
			&=& \arg \max _\mathbf{x} \triangle_{\mathbf{x}^+}\cdot \mathbb{P}(\textsc{Tn}|\mathbf{x}) -\triangle_{\mathbf{x}^-}\cdot \mathbb{P}(\textsc{Fn}|\mathbf{x}) 
		\end{eqnarray}
	\end{proof}
\end{theorem}


\end{document}